%% file: brs-guarantee.tex
\documentclass[11pt]{article}

\input{macros}

\renewcommand{\tV}{\widetilde{V}}
\renewcommand{\cP}{\mathcal{P}}

\renewcommand{\cL}{\mathcal{L}}
\renewcommand{\cS}{\mathcal{S}}
\newcommand{\grank}[1]{{\rm rank}_{\geq 1 - #1}}

\title{Approximating CSPs with Outliers}

\author{ Suprovat Ghoshal \\University of Michigan \\ suprovat@umich.edu
\and Anand Louis \\ Indian Institute of Science \\ anandl@iisc.ac.in }
\date{}

\begin{document}

\begin{titlepage}
\maketitle
\begin{abstract}
Constraint satisfaction problems (CSPs) are ubiquitous in theoretical
computer science. We study the problem of \strongcsp s, i.e. instances where 
a large induced sub-instance has a satisfying assignment. More formally, 
given a CSP instance $\cG(V, E, [k], \set{\Pi_{ij}}_{(i,j) \in E})$ consisting of a
set of vertices $V$, a set of edges $E$, alphabet $[k]$, a constraint 
$\Pi_{ij} \subset [k] \times [k]$ for each $(i,j) \in E$, 
the goal of this problem is to compute the largest subset $S \subseteq V$ such that
the instance induced on $S$ has an assignment that satisfies all the constraints.

In this paper, we study approximation algorithms for \uniquegames~and related problems under the \strongcsp~framework when the underlying constraint graph satisfies mild expansion properties. In particular, we show that given a \stronguniquegames~instance whose optimal solution $S^*$ is supported on a regular low threshold rank graph, there exists an algorithm that runs in time exponential in the threshold rank, and recovers a large satisfiable sub-instance whose size is independent on the label set size and maximum degree of the graph. Our algorithm combines the techniques of Barak-Raghavendra-Steurer (FOCS'11), Guruswami-Sinop (FOCS'11) with several new ideas and runs in time exponential in the threshold rank of the optimal set. A key component of our algorithm is a new threshold rank based spectral decomposition, which is used to compute a ``large'' induced subgraph of ``small'' threshold rank; our techniques build on the work of Oveis Gharan and Rezaei (SODA'17) and could be of independent interest.
\end{abstract}

\end{titlepage}

\input{introduction}

\input{overview}

\input{prelims}

\input{find-expander}
\input{thm-main}

\input{hardness}

\bibliographystyle{alpha}

\paragraph{Acknowledgements.}

AL was supported in part by SERB Award ECR/2017/003296, a Pratiksha Trust Young Investigator Award, and an IUSSTF virtual center on “Polynomials as an Algorithmic Paradigm”.

\bibliography{main}
\appendix
\input{appendix}

\input{vert-sep}

\end{document}

%% file: macros.tex
\usepackage[utf8]{inputenc}
\usepackage{style}
\usepackage{mdframed}
\usepackage{mdwlist}
\usepackage{nicefrac}
\usepackage[parfill]{parskip}
\usepackage{setspace}
\usepackage[ruled,linesnumbered]{algorithm2e}
\usepackage{amsmath}
\usepackage{amssymb}
\usepackage{complexity}

\newcommand{\Ex}{\E}

\renewcommand{\cP}{{\cal P}}
\newcommand{\tright}{\blacktriangleright}
\allowdisplaybreaks

\newcommand{\defeq}{\overset{\rm def}{=}}
\usepackage{microtype}

\newtheorem{problem}[theorem]{Problem}
\usepackage{thm-restate}

\usepackage{boxedminipage}

\newcommand{\paren}[1]{\left(#1 \right )}

\newcommand{\ceil}[1]{\lceil #1 \rceil}

\newcommand{\smallsetexpansion}{{\sc SmallSetEdgeExpansion}}

\newcommand{\stronguniquegames}{{\sc StrongUniqueGames}}
\newcommand{\uniquegames}{{\sc UniqueGames}}
\newcommand{\oddcycletransversal}{{\sc OddCycleTransversal}}

\newcommand{\tcol}{$3$-{\sc Coloring}}
\newcommand{\strongcsp}{{\sc Strong-CSP}}
\newcommand{\lbl}{{\sc LabelCover}}
\newcommand{\mcsp}{{\sc Max-CSP}}
\newcommand{\balvertsep}{{\sc BalancedVertexSeparator}}

\newcommand{\tbigO}{\tilde{\mathcal{O}}}
\newcommand{\tbigo}[1]{\tbigO\left(#1\right)}

\newcommand{\wh}[1]{\widehat{#1}}

\newcommand{\non}{\nonumber}


%% file: introduction.tex
\section{Introduction}

An instance of a $2$-Constraint Satisfaction Problem ($2$-CSP)
$\cG(V, E, [k], \set{\Pi_{ij}}_{(i,j) \in E})$ consists of a
set of vertices $V$, a set of edges $E$, alphabet $[k]$, and a constraint 
$\Pi_{ij} \subseteq [k] \times [k]$ for each $(i,j) \in E$.
The goal of this problem is to compute an assignment $f:V \to [k]$
such that the fraction of constraints satisfied is maximized;
this optimal fraction is also called the {\em value} of this instance, and is formally denoted by ${\rm Val}(\cG)$.
Many common optimization problems such as Max Cut, Unique Games,
Graph Coloring, $2$-SAT, etc. are $2$-CSPs. Designing approximation algorithms for specific CSPs are central problems in
the study of algorithms and have been studied extensively, for e.g.,
Max-Cut \cite{gw94}, Unique Games \cite{cmm06a,cmm06}, etc. 
There is also a long line of work which deal with algorithms for general CSPs (see \cite{rag08,rs09,BRS11,GS11}).

A particular parameter regime of interest is when the CSP instance is ``almost'' fully satisfiable.
There are several ways for quantifying this, one of which is by
asking the value of the CSP instance be close to $1$.  
This can also be viewed as the setting where deleting a small number of edges from the instance
results in an instance that is fully satisfiable.
There has been extensive work on designing algorithms for CSPs in this regime;
we give a brief survey in Section \ref{sec:related}.
Another way a CSP can be almost satisfiable 
is if a small number of {\em outlier} vertices can be deleted (all the edges incident
on these vertices would also be deleted) to obtain an instance 
which is fully satisfiable. 	
The main focus of our work is to study algorithms for CSPs in this model;
we define it below formally.

\begin{problem}[\strongcsp]				
\label{prob:strcsp}
Given an instance $\cG(V, E, [k], \set{\Pi_{ij}}_{(i,j) \in E})$ consisting of a
set of vertices $V$, a set of edges $E$, alphabet $[k]$, and a constraint 
$\Pi_{ij} \subseteq [k] \times [k]$ for each $(i,j) \in E$, 
compute the largest $S \subseteq V$ such that
the instance induced on $S$ has value $1$.
\end{problem}

We refer to an optimal set of vertices for Problem \ref{prob:strcsp} as {\em good vertices}\footnote{Note that such a set of vertices may not be unique, in which case, we will fix such a collection of vertices, and call it the set of good vertices.}, and denote them by $V_{\rm good}$. A naturally arising such instantiation of \strongcsp's is the \oddcycletransversal~problem. Here, given a  graph $G = (V,E)$ as input, the objective is to delete the smallest fraction of vertices so that the graph induced on the remaining vertices is bipartite. This is easily seen as an instance of a \strongcsp~--  here the predicate on the edges is the ``Not Equals'' predicate on the label set $\{0,1\}$. \oddcycletransversal~is a well studied problem. In general, it is known be constant factor inapproximable~\cite{BK09} (assuming the Unique Games Conjecture), and the best known upper bounds (in terms of fraction of vertices deleted) are $O(\delta\sqrt{\log|V|})$~\cite{ACMM05} and $O(\sqrt{\delta \log d})$~\cite{GL20} -- where $\delta$ is the optimal fraction of vertices to be deleted and $d$ is the maximum degree of the graph -- the latter bound is also tight upto constant factors assuming the Unique Games Conjecture \cite{GL20}. Given these worst case bounds, one might ask if there are natural classes of instances under which \oddcycletransversal~admits better approximation?

For the specific setting of \oddcycletransversal, there are several such classes which exhibit improved approximation guarantees. For instance, for the setting of planar graphs, the natural linear programming relaxation is known to be exact~\cite{FMAU92}, and therefore admits an exact polynomial time algorithm. Furthermore, for $K_r$-minor closed graphs, Alev and Lau~\cite{AlevLau17} gave an $O(r)$-approximation algorithm. On the other hand, since \oddcycletransversal~is fixed parameter tractable with respect to treewidth~\cite{LMS11}, it admits exact polynomial time algorithms for graphs with bounded treewidth. Note that these also happen to be characterizations which end up implying easy instances for Max-CSPs. Motivated by this connection, we investigate whether there are spectral characterizations under which \oddcycletransversal~(and more generally, \strongcsp's) admit improved approximation. In particular, we study instances which are expanding, or more generally, have low threshold rank{\footnote{It is folklore that a graph can be a small-set-expander if and only if it has bounded number of large eigenvalues. For a more quantitative version of this statement, see Theorem \ref{thm:h-cheeger}.}}. Formally, the threshold rank of a graph is defined as follows.

\begin{definition}[Threshold rank]
	\label{def:trank}
	Given an undirected graph $G = (V,E)$, let $A$ denote its weighted adjacency matrix $G$ and let $D$ denote the diagonal matrix where 
	$D(i,i)$ is the weighted degree of vertex $i$. 
	The $(1 - \epsilon)$ threshold rank of $G$, denoted by $\grank{\epsilon}(G)$
	is defined as the number of eigenvalues of $D^{-\frac{1}{2}} A D^{-\frac{1}{2}}$ that are greater than or equal to
	$1 - \epsilon$.
\end{definition}

In the setting of CSPs, low threshold rank instances have been studied extensively -- the study of such instances was instrumental in the  development of sub-exponential time algorithms for \uniquegames~and \smallsetexpansion~\cite{Kol10,ABS15,BRS11}. In particular, for the edge deletion analogue of \oddcycletransversal~i.e., {\sc Max-Cut},~\cite{BRS11} gave a $(1/\lambda_t)$-approximation algorithm running in time $n^{{\rm poly}(t)}$, where $\lambda_t$ is the $t^{th}$ smallest eigenvalue of the normalized Laplacian. Surprisingly, to the best of our knowledge, no such analogous results are known for \oddcycletransversal. Furthermore, random instances of CSPs are expanding, and naturally have low threshold rank. This motivates us to explore the approximability of \oddcycletransversal~and other \strongcsp's in low threshold instances. In fact, we study them under the more stringent setting where only the graph induced on good vertices (constituting the fully satisfiable sub-instance) is assumed to have low threshold rank, as opposed to the full graph having low threshold rank.

{\bf Max-CSPs vs. StrongCSPs}. This relaxation, in addition to making the setting more challenging, is also motivated by our wish to exhibit a separation between the approximability of edge deletion and vertex deletion  problems, i.e., namely {\sc Max-CSP}s and \strongcsp's. We point out that under an identical setting (where only a $(1 - \delta)$-sized subset has low threshold rank), {\sc Max-CSP}s can be arbitrarily hard to approximate. Indeed, consider a {\sc Max-CSP} instance where the $(1 - \delta)$-sized subset $V_{\rm good}$ induces a constant degree expander with trivially satisfiable constraints, and the edges going across $V_{\rm good},V^c_{\rm good}$ encode a denser hard to approximate {\sc Unique Game} instance with large gap and larger vertex degrees. It is easy to see that such instances do not admit efficient constant factor approximation guarantees with respect to the edge satisfaction objective i.e., that of finding an assignment that satisfies the maximum fraction of constraints. On the other hand, our results in the current work show that the same instances when interpreted as \strongcsp's are easy (i.e, with respect to the vertex deletion objective, see Definition \ref{prob:strcsp}).  Therefore, it is not immediately obvious that conditions under which {\sc Max-CSP}s are easy also translate to conditions under which \strongcsp's might be easy and vice versa, and hence, the broader agenda of identifying clean characterizations under which there is a separation in the approximability of the two classes of problems might yield useful insights towards understanding the limitations of the approximation techniques for problems from either class.

{\bf Connection to Fortification}. A final motivation for studying \strongcsp's in the above setting is that the problem of finding slightly smaller sub-instances with better ``local'' approximation guarantees is closely related to notion of {\em fortification}. Informally, a {\sc Max-CSP} instance is said to be fortified if every large sub-instance of the CSP has (relative) optimal value no larger than the global optimal. Fortification is widely studied in the context of parallel repetition~\cite{Mos14,BVY17,Mos21}, and in particular, recent works~\cite{Mos21} show that fortified {\sc Unique Game} instances with hypercontractive small set expansion profiles can be used to bypass bottlenecks towards establishing {\em strong parallel repetition} for {\sc Unique Game} instances. Consequently, this reduces the task of establishing UGC to that of showing that a family of fortified Boolean CSPs on small-set-expanders are hard. Given that \strongcsp's can be thought as deciding whether an instance is fortified (in the perfect completeness regime), and the tight connections between small-set-expansion and threshold rank (e.g., \cite{ABS15},\cite{LOT14},\cite{LRTV12}), these considerations further motivate the study of \strongcsp's even in the simpler setting where the full underlying constraint graph has low threshold rank. Motivated by the above considerations, we study the \stronguniquegames~and related problems in this setting:

\begin{problem}[\stronguniquegames]
\label{prob:strug}
Given an instance $\cG(V, E, [k], \set{\Pi_{ij}}_{(i,j) \in E})$ consisting of a
set of vertices $V$, a set of edges $E$, alphabet $[k]$, and a bijection
$\pi_{ij} \subseteq [k] \times [k]$ for each $(i,j) \in E$, 
the goal of this problem is to compute the largest $S \subseteq V$ such that
the instance induced on $S$ has value $1$.
\end{problem}
%

The \stronguniquegames~problem is a natural variant of \uniquegames, it and its variants express several well studied problems such as \oddcycletransversal, among others. There have been extensive work on the above problems, see Section \ref{sec:related} for a detailed review. Our main results in this paper are improved approximation algorithms for these problems in the setting where the induced graph on the good vertices has low threshold rank.

\subsection{Our Results}

Our main result is a new approximation algorithm for the \stronguniquegames~problem where the induced sub-graph on the satisfiable set has low threshold rank. In order to make the theorem statements concise, we will define the notion of a subset being $\lambda^*$-good.

\begin{definition}[$\lambda^*$-good]
	Given a CSP constraint graph $G = (V,E)$, a subset $V^* \subseteq V$ is said to be $\lambda^*$-good if the following conditions hold.
	\begin{enumerate}
			\item $\grank{\lambda^*}\left(\cG[V^*]\right) \leq (1/\lambda^*)^{10}$. \footnote{The constant $10$ in the exponent is arbitrary, and can be chosen to any large constant $C$, at the cost of loss in ${\rm poly}(C)$-multiplicative factors in the fraction of vertices deleted by the algorithm. We instantiate it to be $10$ for ease of notation.}
			\item $G[V^*]$ is regular.
	\end{enumerate}
\end{definition}

The above is a quantitative characterization of induced low-rank instances studied in this paper -- all of our results are based on the above setting. Our first result is for \stronguniquegames~instances with small vertex induced low threshold rank, as stated in the following theorem.

\begin{restatable}{rethm}{strongug}
	\label{thm:strong-ug}
	Let $\delta, \lambda^* \in (0,1)$ be such that $\delta \leq (\lambda^*)^{100}$. Let $\cG(V,E,[k],\{\pi_e\}_{e \in E})$ be a \stronguniquegames~instance such that there exists\footnote{We do not assume that such a set is unique, we just need the existence at least one such subset.} a $\lambda^*$-good subset $V_{\rm good}$ of size at least $(1 - \delta)n$ such that ${\rm Val}\paren{\cG[V_{\rm good}]} = 1$. Then there exists a randomized algorithm that runs in time $n^{{\rm poly}(k/\delta)}$ and outputs a subset $\tV \subseteq V$ of size at least $(1 - \delta^{1/12})n$ and a partial labeling $\sigma: \tV \to [k]$ such that $\sigma$ satisfies all induced constraints in $\cG[\tV]$.
\end{restatable}	

The above theorem illustrates the tractability of the \stronguniquegames~problem in the setting where just the instance induced on the satisfiable set has low threshold rank. To put the above result in perspective, \cite{GL20} showed that given a \stronguniquegames~instance with value $(1 - \delta)$, it is Unique Games hard to output a subset of relative size $(1 - \Omega(\sqrt{\delta \log d \log k}))$, where $d$ is the maximum degree of the graph and $k$ is the label set size. We remark that the exponent in the fraction of vertices deleted (i.e., $\delta^{1/12}$) might be improvable and we have not made further attempts towards optimizing it. Theorem \ref{thm:strong-ug} almost directly leads to quantitatively similar results for the 
\oddcycletransversal~and \balvertsep~problems, stated as corollaries. 

\begin{corollary}		\label{corr:oct}
	Let $\delta, \lambda^* \in (0,1)$ be such that $\delta \leq (\lambda^*)^{100}$. Let $\cG = (V,E)$ be a graph for which there exists a $\lambda^*$-good subset $V_{\rm good} \subseteq V$ of size at least $(1 - \delta)n$ such that $G[V_{\rm good}]$ is bipartite. Then there exists an algorithm which runs in times $n^{{\rm poly}(1/\delta)}$ which outputs a set $V' \subseteq V$ of size at least $(1 - \delta^{1/12})n$ such that $G[V']$ is bipartite.
\end{corollary}

\begin{restatable}{rethm}{vertsep}
	\label{thm:sep}
	Let $\delta, \lambda^* \in (0,1)$ be such that $\delta \leq (\lambda^*)^{100}$. Let $\cG = (V,E)$ be a graph for which there exists a $\lambda^*$-good subset $V_{\rm good} \subseteq V$ of size at least $(1 - \delta)n$ such that the following holds. There exists a partition $V_{\rm good} = A \uplus B$ such that $E_{G[V_{\rm good}]}(A,B) = \emptyset$ i.e, $A$ is disconnected from $B$ in $G[V_{\rm good}]$. Then there exists a randomized algorithm which runs in time $n^{{\rm poly}(1/\delta)}$ and outputs a set $S$ of size at most $O(\delta^{1/12}n)$ and a partition $A',B'$ of $V \setminus S$ such that (a) $E_{G[V \setminus S]}(A,B) = \emptyset$ and (b) $(\gamma - \delta^{1/12})n \leq \min(|A'|,|B'|) \leq (\gamma + \delta^{1/12})n$ where $\gamma = \min(|A|,|B|)/n$. 
\end{restatable}	

For both \oddcycletransversal~as well as~\balvertsep, the best known approximation algorithm for general instances have an approximation guarantee of $O(\sqrt{\log |V|})$~\cite{FLH08,ACMM05}. Furthermore,~\cite{GL20} showed that given a $(1 - \delta)$-satisfiable instance of \oddcycletransversal, assuming UGC, it is \NP-Hard to find set of size $(1 - \Omega(\sqrt{\delta \log d}))$ which induces a bipartite graph. {\em It is important to note that our results hold for more restrictive setting where we assume the low threshold rank guarantee on the good set}. In particular, the technical core of our results is a spectral decomposition theorem which can be used to find a large subset that induces a sub-graph with relatively small threshold rank. We state an informal version of it here for reference.

\begin{theorem}[Informal version of Theorem \ref{thm:thr-graph}]			\label{thm:thr-inf}
	The following holds for every $0 < \delta \leq 0.1$. Let $G = (V,E)$ be a $d$-regular graph on $n$-vertices such that there exists a set $V_{\rm good} \subseteq V$ of size at least $(1 - \delta)n$ such that ${\rm rank}_{\geq 1 - \delta^{0.1}}(G[V_{\rm good}]) \leq K$. Furthermore, suppose $K \leq 1/\delta^{100}$. Then there exists an efficient algorithm outputs a set $V'' \subseteq V$ of size at least $(1 - O(\delta^{1/10}))n$ such that $\grank{\delta^{0.1}}(G[V'']) \leq  {\rm poly}(1/\delta)$. Moreover, the subset $V''$ itself is a disjoint union of constant number of $\Omega(n)$-sized subsets, each of which induces an expander.
\end{theorem}

The above decomposition result adds to the already extensive literature on spectral decomposition -- however, the above decomposition result is incomparable in terms of its setting and guarantees to the ones existing in the literature. For comparison, we describe the two previous such results which are closest in terms of the setting and the guarantees:
\begin{itemize}
	\item In \cite{ABS15}, Arora, Barak and Steurer show that any $n$-vertex graph can be decomposed into non-expanding subsets which induce sub-graphs of $(1 - \epsilon^5)$-threshold rank at most $n^\epsilon$. While their result does not require the graph to contain a large low threshold rank sub-graph, their decomposition result can only guarantee a substantially weaker threshold rank bound of $n^\epsilon$ (as opposed to the constant bounds guaranteed in Theorem \ref{thm:thr-inf}). We clarify that their $n$-dependent bound on the threshold rank is indeed unavoidable, since they make no assumptions on the threshold rank structure of the graph~\cite{MS18}.
	\item In \cite{GR17}, Oveis Gharan and Rezaei show that given a regular graph which contains a $\kappa n$-sized spectral expander, one can efficiently find subset of size at least $3\kappa n/8$ with spectral gap multiplicatively comparable to that of the optimal induced expander. Again, their result is not directly comparable to ours since even in graphs which contain a $(1 - \delta)n$-sized induced expander, their algorithm is only guaranteed to output a $3(1 - \delta)n/8$-sized subset which induces an expander. In comparison, for similar instances, Theorem \ref{thm:thr-graph} guarantees a $(1 - \delta^{0.1})$-sized subset which induces a ``low threshold rank graph'' -- which itself is guaranteed to be a union of linear sized expanders. On the other hand, our result only applies in the setting  $\kappa \to 1$, whereas their result holds for any constant $\kappa \in (0,1)$.
\end{itemize} 

We point out that our actual spectral decomposition theorem (Theorem \ref{thm:thr-graph}) differs from the informal version stated above (i.e., Theorem \ref{thm:thr-inf}) in a couple of crucial ways. Firstly, we only assume that only the underlying good graph $G[V_{\rm good}]$ is regular (as opposed to the full graph being regular) and make no assumptions on the degree distribution of the set of outlier vertices $V \setminus V_{\rm good}$ -- indeed, these assumptions allow us to include instances which show a separation between the approximability of the {\sc Max-CSP} and \strongcsp~objectives. Secondly, our actual guarantee is slightly more robust in the following sense: given any $(1 - \delta^{O(1)})$-sized subset $V' \subseteq V$ (where $V' \not\subset V_{\rm good}$), one can find another subset $V'' \subseteq V'$ of size $(1 - \delta^{O(1)})$ such that ${\rm rank}_{1 - \delta^{O(1)}}(G[V'']) \leq {\rm poly}(1/\delta)$. The structural fact that we can still recover a large low threshold rank subgraph within any large subset $V'$ is interesting on its own, we are not aware of similar results in the previous literature on spectral decomposition.

\begin{remark}[On the regularity assumption]
	We point out that our threshold rank decomposition result, and more generally the approximation guarantees from Theorem \ref{thm:strong-ug} and its corollaries also hold as is as long as $V_{\rm good}$ is $\lambda^*$-good and is contained in any subset $\tilde{V}$ (where $\tilde{V}$ may strictly contain $V_{\rm good}$) for which $G[\tilde{V}]$ induces a regular subgraph -- this naturally subsumes the more commonly studied setting where the full graph has low threshold rank and is regular~\cite{ABS15,BRS11}. As in these works, our results will also hold for the setting where the graph is non-regular; in that setting, the guarantees of the threshold decomposition result and our algorithm will involve bounds on the volume of the subset deleted by the algorithm (as opposed to bounds on the size of the subset).
\end{remark}	

\subsubsection*{Hardness of \strongcsp's} 
Given our algorithmic results hold for structured instances i.e., the subgraph induced by the good set has low threshold rank, an immediate question is if it is possible to obtain quantitatively similar approximation guarantees without making any assumptions. Towards that, our first observation is that arbitrary Strong $2$-CSPs can be almost polynomially hard to approximate, as stated by the following fact.

\begin{observation}[Hardness of General Strong $2$-CSPs]
		The following holds for any small $\epsilon > 0$. Given a $2$-CSP $\Psi(V,E,\{\psi\}_{e \in E})$ over label set $\{0,1\}$, it is $\NP$-Hard to find a subset $V' \subseteq V$ of size $|V'| \geq n^{1-\epsilon}|V^*|$ such that all induced constraints on $V'$ are satisfiable. Here $V^*$ is a set of largest cardinality for which there exists a labeling which satisfies all the induced constraints on $V^*$.
\end{observation}

The fact follows simply by using the observation that the Maximum Independent Set problem can be modeled as \strongcsp~on label set $\{0,1\}$ with arity $2$ (see Appendix \ref{sec:ind-set-hard} for a formal explanation). On the other hand, it is known that all general $2$-CSPs admit constant factor approximation (when the label set size is a constant). For e.g., for any $2$-CSP on $\{0,1\}$ just a random assignment itself satisfies at least $1/4$-fraction of constraints in expectation. This shows that \strongcsp's can be strictly harder that \mcsp s. Clearly, one can expect general \strongcsp's to only get harder for larger arities, so we choose to relax the requirements of \strongcsp's and ask the following question. Consider a \mcsp~which is known to be hard to approximate to a factor of $\alpha$. Then it is natural to ask, if given such an instance, can we delete a few vertices, and then output a labeling on the remaining instance which has approximation factor strictly better than $\alpha$. The following theorem answers the question in the negative for the specific setting where the CSP is Max-$4$-Lin.

\begin{theorem}					
\label{thm:4lin}
The following holds for any constants $\alpha,\eta,\nu \in (0,1)$. Given a system of equations $\Psi$ of arity $4$, on variables $X_1,X_2,\ldots,X_n$ taking values in $\mathbbm{F}_2$, it is \NP-Hard to distinguish between the following cases 
\begin{itemize}
	\item There exists an assignment to the variables which satisfies at least $(1 - \eta)$-fraction of constraints in $\Psi$.
	\item No subset $S \subseteq V$ of size at least $\alpha n$ induces a system of equations for which there exists an assignment which satisfies at least $(1/2 + \nu)$ fraction of the induced constraints.
\end{itemize}
\end{theorem}
The above can be thought of as an instance of approximation resistance in a \strongcsp~sense; it is a strengthening of $(1/2 + \nu)$-inapproximability for Max-$3$-Lin shown by H\r{a}stad in the seminal work~\cite{Has01}. We prove the above hardness result by combining the techniques from \cite{Has01} with novel application of expansion properties of the inner and outer verifiers. In particular, Theorem \ref{thm:4lin} says that one cannot hope to do slightly better than its inapproximability factor (which is matched by the naive random guessing algorithm) on any smaller sub-instance for approximation resistant predicates. 

\subsection{Related Work}
\label{sec:related}

\paragraph{Strong Unique Games}
Ghoshal and Louis~\cite{GL20} gave an algorithm that takes as input
 an instance of \stronguniquegames~having a set of size $(1 - \epsilon)n$
such that the instance induced on that set has value $1$, and outputs  a set 
of size at least $\paren{1 - \tbigo{k^2} \epsilon \sqrt{\log n}} n$ such 
that instance induced on the set has value $1$.
They gave another algorithm that produced a set of size 
$\paren{1 - \tbigo{k^2} \sqrt{\epsilon \log d}} n$ such that the instance instance on that set has value $1$,
where $d$ is the largest vertex degree of the instance. 
They also showed that it is Unique Games hard (in certain regimes of parameters) to compute a set of size
larger than $1 - O(\sqrt{\epsilon \log d \log k})$ such that the induced instance on 
this set is satisfiable. 
Their work
showed the connection between \stronguniquegames~and small-set vertex expansion
in graphs, and used the machinery (hypergraph orthogonal separators) developed in 
the context of approximation 
algorithms for small-set vertex expansion in graphs and hypergraph small-set expansion
\cite{lm16} in obtaining their approximation algorithms.

\paragraph{General CSPs.}

There have been several works which give approximation algorithms for $2$-CSPs. \cite{AKKSTV08} were the first to study \uniquegames~in the setting where the underlying constraint graph is an expander; they gave an algorithm with the approximation factor depending on only the second largest eigenvalue of the normalized Laplacian matrix of the instance. Subsequent works such by Barak, Raghavendra and Steurer~\cite{BRS11} and Guruswami and Sinop~\cite{GS11} extended this framework to general $2$-CSPs when the underlying constraint graph and the label extended graph have low threshold rank respectively, with the algorithms running time exponential in threshold rank. On the other hand, Kolla~\cite{Kol10} gave spectral approximation algorithms for \uniquegames~and \smallsetexpansion. Building on this, Arora, Barak and Steurer~\cite{ABS15} gave sub-exponential time algorithms for \uniquegames~and \smallsetexpansion. In a recent work, \cite{BBKSS20} give efficient algorithms for \uniquegames~based on the Sum Of Squares (SoS) hierarchy, when the underlying constraint graph is an SoS certifiable small set expander.


\paragraph{Graph Partitioning and CSPs with Cardinality Constraints}

Graph partitioning with vertex/edge expansion objectives has been extensively studied under the lens of approximation algorithms. Feige, Lee and Hajhiyaghayi~\cite{FLH08} and Louis, Raghavendra and Vempala~\cite{LRV13} give approximation algorithms for finding small size balanced vertex separators and minimizing vertex expansion respectively. Guruswami and Sinop \cite{GS11,GS13} gave improved approximation algorithms for several graph partitioning problems dealing with edge expansion for low threshold rank instances. \cite{LV18} studied a planted model of instances where the graph induced on either side of the planted cut satisfies a lower bound requirement on its spectral gap in addition to satisfying some other properties; they gave exact and constant factor bi-criteria approximation algorithms for balanced vertex expansion for various ranges of parameters. They also gave a constant factor bi-criteria approximation algorithm for balanced vertex expansion for instances where one side of the optimal cut has a subgraph on $\Omega(n)$ vertices satisfying a lower bound requirement on its spectral gap. \cite{LV19} gave some similar results for $k$-way edge expansion and $k$-way vertex expansion.

The problem of decomposing a graph into expanders is also a well studied problem and has several applications to approximation algorithms. In \cite{Tre05}, Trevisan gave a decomposition of a graph into non-expanding set which induce expanders. There have been several subsequent works \cite{ABS15,GT13,GT14} which deal with the problem of partitioning a graph into expanding/low threshold rank graphs. Oveis Gharan and Rezeai \cite{GR17} study the problem of finding a large subset of vertices such that the graph induced on them is an expander; we discuss this more in Section \ref{sec:overview}.

%% file: overview.tex
\section{Overview and Techniques}
\label{sec:overview}

We begin by reviewing the by now standard {\em Propagation Rounding} based framework which was introduced informally in \cite{AKKSTV08} and then later developed in \cite{BRS11,GS11}. For simplicity, we shall restrict our discussion to the setting of \uniquegames. Consider the following convex program which is the $R$-level Sum-of-Squares (SoS) lifting of SDP relaxation for \uniquegames:
\begin{equation}				\label{eqn:ug-relax}
	\min_{\substack{\mu \textnormal{ is a degree-}R \\ \textnormal{pseudo-distribution\footnotemark}}} \Ex_{(i,j) = e \sim E} \Pr_{(X_i,X_j) \sim \mu}\left[X_i \neq \pi_{j \to i}(X_j)\right].
\end{equation}
\footnotetext{Informally, a degree-$R$ pseudo-distribution is a collection of local distributions $\{\mu_{S}\}_{S}$ for every subset $S \subseteq V$ of size at most $R$, which are pairwise consistent up to all variables sets of size at most $R$ (see Section \ref{sec:lasserre} for more details).}
The above convex program is intended to minimize the number of unsatisfied edges by the (pseudo)-distribution. The algorithm proceeds along the following steps.

\begin{enumerate}
	\item Solve the $R$-round Lasserre relaxation for the SDP where $R$ is chosen large enough as a function of the error to be tolerated, and the threshold-rank of the instance. Let $\mu := \{\mu_{S,\alpha}\}$ be the degree-$R$ pseudo-distribution corresponding to the optimal value of the relaxation.
	\item Choose a subset $S$ appropriately, sample an assignment $x_S$ to the variables in $S$ from the local distribution $\mu_S$.
	\item Label the remaining vertices $i \in V \setminus S$ by sampling from their respective conditional distributions $\mu_{i|x_S}$ independently. 
\end{enumerate} 

The main idea used in the aforementioned works for relating the expected value of the rounded solution to the SDP objective is the so called {\em local-to-global} correlation property \cite{BRS11,GS11}, which has the following key consequence. If the underlying constraint graph has constant threshold rank, then conditioning on constant levels of the SoS solution should result in pseudo-distributions that have small average local correlation i.e.,
\[
\Ex_{(i,j) \sim E}\left[{\rm Corr}_{\mu|x_S}(X_i,X_j) \right] \leq o(1).
\] 
Consequently, independent sampling from the marginals of conditional pseudo-distribution $\mu|x_S$ will results in labelings that which have value close to optimal of the lifted SDP. While this recipe and its variants has been remarkably successful in dealing with {\sc Max-CSP}s \cite{BRS11,GS11,BBKSS20}, it is easy to see that the this framework does not translate well to the framework of \strongcsp's studied in this paper, as we briefly describe below. 

Firstly, note that in the setting of \strongcsp's, the emphasis is on deleting vertices to ensure that all surviving constraints are simultaneously satisfiable. This is in direct contrast to the aforementioned results where the algorithms are allowed to output labelings which satisfy ``almost all'', but not necessarily, ``all'', constraints. A naive approach towards extending the above to our setting would be to first find a good labeling that satisfies almost all edges, and then delete the vertices corresponding to the violated edges. However, doing so might result in approximation guarantees that are worse by a factor of the max-degree. Furthermore, this approach can fail badly in instances where the induced sub-instance on the good vertices $\cG[V_{\rm good}]$ is sparse (i.e, constant degree), but the full graph is relatively dense and almost non-satisfiable, since these algorithms are designed to compete against the global optimum for the edge-satisfaction version of the problem. A final hurdle is that the local-to-global correlation guarantee, which was the key property used to guarantee the goodness of the rounding algorithm, might not hold for the full constraint graph of $\cG$ since in our setting, the constant threshold-rank guarantee may only hold for the constraint graph induced on $V_{\rm good}$. In fact,  the threshold rank of the full graph can be as large as $\Omega(|V \setminus V_{\rm good}|)$ which implies that conditioning on constant levels of the SoS solution might result in pseudo-distributions that don't guarantee any local-to-global correlation like property. These issues taken together guide our approach to the design of our algorithm (described informally in Figure \ref{fig:alg-inf}); we describe and motivate the various steps of the algorithm details in the remainder of this section.

\begin{figure}[ht!] 	
\begin{mdframed}
	{\bf Input}: A \uniquegames~instance $\cG(V_\cG,E_{\cG},[k],\{\pi_e\}_{e \in E})$ satisfying the conditions of Theorem \ref{thm:strong-ug}. \\
	{\bf Algorithm}:
	\begin{itemize}
	\item[$\tright$] {\bf Threshold Rank Decomposition} As a first step, we compute a $(1 - O(\delta^c))$-sized subset $V''$ of $V_\cG$ with low threshold-rank and bounded degree using our spectral decomposition algorithm (Theorem \ref{thm:thr-graph}).  
	\item[$\tright$] {\bf SDP with Slack Variables}. We solve the $R$-level SoS relaxation of a modified SDP for \uniquegames~instance induced on $V''$ with the extended label set $[k] \cup \{*\}$, where the label $*$ is meant to indicate vertices which are to be deleted. 
	\item[$\tright$] {\bf Low Variance Rounding}. We sample an assignment $\alpha$ for an appropriately chosen subset $S$, and then label the vertex $i \in V$ with the label with the largest probability in the conditional marginal $\mu_{i|X_S = \alpha}$.
\end{itemize} 
\end{mdframed}
\caption{StrongUG-Informal}
\label{fig:alg-inf}
\end{figure}

\paragraph{Finding a large bounded-degree low threshold-rank graph.}

Since the full instance in our setting can have arbitrarily large threshold rank (due to the edges incident on the set of outlier vertices), a natural way to overcome this issue would be to zoom into a large (i.e, $(1 - o_\delta(1)$-sized) subset of vertices which induces a subgraph with (comparably) low threshold rank. We do this by using a new threshold rank based spectral decomposition algorithm with the following guarantee: given a graph $G = (V,E)$ for which there exists a $(1 - \delta)|V|$ sized subset which induces a sub-graph that is regular and has low threshold rank, the algorithm returns a $(1 - \delta^{O(1)})$-sized subset with threshold\footnote{Here the threshold parameter is dependent on $\delta$ and the optimal value of the \stronguniquegames~instance.} rank at most ${\rm poly}(1/\delta)$. This algorithm is the main technical contribution of this paper; in particular, it
combines a classical approximation algorithm for the partial vertex cover problem and extensions of spectral partitioning primitives from Oveis Gharan and Rezaei~\cite{GR17} to the setting of low threshold rank graphs. We defer a more detailed discussion of this step to Section \ref{sec:overview-part} for now and proceed with our discussion of the subsequent steps of the full algorithm.

\paragraph{Solve SoS relaxation with Slack Variables.}

In the next step, we consider the SDP for \uniquegames~modified with slack variables. Specifically, let $\cG(V,E,[k],\{\pi_{e}\}_{e \in E})$ be the \uniquegames~instance. Due to the above step, we can directly assume that the full graph has low threshold-rank and bounded vertex degrees. Furthermore, since the previous step only removes a tiny fraction of vertices, we can assume that there exists a subset $V_{\rm sat} \subseteq V$ such that  $|V_{\rm sat}| \geq (1 - 2\delta)|V|$ and $\cG[V_{\rm Sat}]$ is fully satisfiable\footnote{Note that $V_{\rm sat}$ may be a strict subset of $V_{\rm good}$ since the previous step may remove a few vertices from $V_{\rm good}$.}. Now given $\cG$, we consider a partial\footnote{The nomenclature ``partial'' Unique Game was introduced in \cite{RS10} and refers to a Unique Game with the additional property that a fixed fraction of vertices are allowed to be left as unlabeled.} Unique Game $\cG'(V,E,[k] \cup \{*\},\{\Pi_e\}_{e \in E})$ with the global constraint that the fraction of vertices that can be labeled $'*'$ is at most $2\delta$. Here the label $*$ is meant to indicate vertices that are supposed to be deleted. Consequently, for any edge $e \in E$, we define the extended constraint set $\Pi_e = \pi_e \cup (\{*\} \times \Sigma) \cup (\Sigma \times \{*\})$. Note that this constraint is no longer a ``unique game'' constraint. The final SoS relaxation used is almost identical to Eq. \ref{eqn:ug-relax}, along with the following modifications:

\begin{itemize}
	\item[$C_1$:] The pseudo-distribution is now over assignments to variables from the extended label set $[k] \cup \{*\}$.
	\item[$C_2$:] We add global cardinality constraint $\Pr_{i \sim V}\Pr_{\mu_i | X_s = \alpha}[X_i = * ] \leq 2\delta$, for every subset $S$ of size at most $R$ and assignment $\alpha \in ([k] \cup \{*\})^S$.
	\item[$C_3$:] We also add the constraint 
		\[ \Pr_{\mu_{ij}|X_S = \alpha}[\pi_{i \to j}(X_i) \neq X_j] \leq \Pr_{\mu_i|X_S = \alpha}[X_i = *] + \Pr_{\mu_j |X_S = \alpha}[X_j = *] \] 
		for every edge $(i,j) \in E$, subset $S$ and corresponding assignment $\alpha$.
\end{itemize}

The cardinality  constraint $(C_2)$ is intended to ensure that conditioned on any assignment that is assigned a non-zero probability mass by the SDP solution, the fraction of vertices that are labeled $*$ under the resulting conditional distribution is at most $\delta$. The edge violation constraints $(C_3)$ are intended to ensure that an edge constraint is allowed to be violated only when one of the end points is labeled $*$. It is easy to verify that this SDP is feasible for $\cG'$. Furthermore, since the previous step guarantees that the max-degree of the surviving graph is at most a constant times the average degree, this implies that the optimal value of the SoS relaxation is at most $\delta^{O(1)}$.
 
\paragraph{Low Variance Rounding.}
In the final step, we have to round the SDP solution to output a large set with the corresponding labeling which satisfies all induced constraints. As mentioned above, the {\em local-to-global} correlation argument in itself is not sufficient for this purpose, as it can only guarantee that a labeling which violates a small fraction of edges. However, it is well known that for certain kinds of CSPs e.g,. \uniquegames, \tcol, the low threshold-rank guarantee implies the stronger property of "{\em conditioning reduces variance}" \cite{BRS11,GS11,AG11}{\footnote{\cite{AG11} actually showed a variant of this statement tailored towards finding large independent sets.}},
which says that for an appropriately chosen subset $S \subseteq V$ we have
\begin{equation}			\label{eq:var-red}
\Ex_{X_S \sim \mu_S} \Big[\Ex_{i \sim V}{\rm Var}\Big[X_i |X_S\Big]\Big] \leq \frac{{\sf SDP}}{\lambda_m},
\end{equation}
whenever $\grank{\lambda_m}(G) \leq m$. Note that this is a strictly stronger property than local-to-global correlation (see Appendix \ref{sec:examp} for an example which separates the two properties). To see why this property is useful in constructing labelings which satisfy all induced constraints, consider a constraint $(i, j) \in E_{\cG}$ such that for some partial assignment $X_S \gets \alpha$, the conditional marginals of vertices $i$ and $j$ have low variance i.e,. ${\rm Var}[X_i |X_S = \alpha] \leq 0.1$ and ${\rm Var}[X_j|X_S = \alpha] \leq 0.1$. Therefore, it follows that there exists labels $a,b \in [k] \cup \{*\}$ for which $\mu_{i = a|X_S = \alpha} \geq 0.9$ and $\mu_{j = b|X_S = \alpha} \geq 0.9$. Furthermore, suppose we assume that $a,b \in [k]$. Then we claim that the SDP constraints imply that $\pi_{i \to j}(a) = b$. This is because for any $(a',b') \in [k] \times [k]$ which violates the edges $(i,j)$, using the edges violation constraints $(C_3)$ and a union bound we get that 
\[
\Pr_{(X_i,X_j) \sim \mu|X_S = \alpha}\Big[X_i = a', X_j = b'\Big] \leq 0.2 .
\]
On the other hand, our choice of labels $a$ and $b$ for vertices $i$ and $j$ (respectively) imply that
\[
\Pr_{(X_i,X_j) \sim \mu|X_S = \alpha}\Big[X_i = a,X_j = b\Big] \geq 1 - \Pr_{X_i \sim \mu|X_S = \alpha}\Big[X_i \neq a\Big] -\Pr_{X_j \sim \mu|X_S = \alpha}\Big[X_j \neq b\Big] \geq 0.8 .
\]  
Therefore, it must be that $\pi_{j \to i}(a) = b$ i.e, the labeling $(a,b)$ satisfies the edge $(i,j)$. In summary, low variance vertices whose leading labels are not $'*'$ induce a satisfiable instance. We point out that a similar observation was also made by Arora and Ge~\cite{AG11} who used it to find large independent sets in low threshold-rank graphs. The above discussion naturally suggests the following rounding process:

\begin{enumerate}
	\item Let $S$ be the subset for which Eq. \ref{eq:var-red} holds. Sample an assignment $\alpha \sim \mu_S$ for $X_S$
	\item Delete the vertices for which ${\rm Var}_{\mu|X_S = \alpha}[X_i] > 0.1$. 
	\item For the remaining vertices $i \in V$, assign the maximum likelihood labeling 
	\[
	\sigma(i) = \argmax_{a \in [k] \cup \{*\}}~~\Pr_{X_i \sim \mu|X_S = \alpha}\Big[X_i = a\Big].
	\]
	\item Delete the vertices labeled as $*$ and output the surviving vertices with the corresponding labeling.
\end{enumerate}

The above discussion ensures that the set output by the rounding scheme is satisfiable. Combining \eqref{eq:var-red} with the SDP bound and the threshold-rank bound established in the previous steps imply that $O(\delta^{O(1)})$ vertices get deleted in step $3$. Furthermore, the global cardinality constraint ensures that the fraction of vertices labeled $'*'$ (and hence deleted) is $O(\delta)$. This with the bound on the vertices deleted in the previous steps imply that the total fraction of vertices deleted is $\delta^{O(1)}$, which concludes the analysis of the algorithm.

\subsection{Threshold Rank based Spectral Partitioning}				\label{sec:overview-part}

As mentioned above, the first step of our algorithm (i.e, the threshold rank decomposition step) is a key technical contribution of this work. Formally, our objective here is the following: given a graph $G = (V,E)$ which contains a $(1 - \delta)$-sized subset $V_{\rm good}$ that induces a regular subgraph with low threshold rank, the objective is to recover a $(1 - o_\delta(1))$-sized subset that has relatively small threshold rank (say ${\rm poly}(1/\delta))$. This in itself is a well motivated question and various versions of it have been studied in the design of approximation algorithms for \uniquegames~and \smallsetexpansion~(see \cite{ABS15} and references therein). However, we  point out that the techniques from these earlier works do not immediately apply to our setting as in these works, the emphasis is rather on finding sub-linear sized sets which induce graphs with threshold rank growing with the number of vertices (with of course, no assumption on the spectrum of the full graph). In our setting, we instead want to design algorithms that exploit the ``almost low threshold'' structure of the instance and output sets that satisfies the stronger guarantees of being almost linear sized and having constant threshold rank. In the remainder of this section, we motivate our design of such an algorithm.

{\bf Finding a Linear Sized Low Rank Set}.  To begin with, let us first consider the simpler setting where we assume that the max-degree of the graph is at most a constant times the degree of the induced good graph $G[V_{\rm good}]$ (we shall later discuss how to achieve this condition at the cost of deleting a few additional vertices). Furthermore, let us first address the even simpler goal of finding a linear (say $n/1000$) sized subset which induces a low threshold rank graph. Again, this in itself is a well motivated problem, and several previous works~\cite{Tre05,GR17} study the related question when the induced subgraph has to be an expander\footnote{Here we refer to any graph whose spectral gap is at least a constant as an expander.}. Most of these works build on the following basic spectral partitioning primitive that also forms the basis of our algorithm:
\begin{gather}				
	\triangleright \textnormal{\it \qquad \qquad There exists an efficient algorithm that given a graph $G = (V,E)$, } \nonumber\\ 
	\textnormal{\it outputs a partition $\cP:=(S,T)$ of $V$ such that either $|S| \geq 3n/4$  and $G[S]$ is an expander,} \non \\
	\textnormal{\it \qquad \qquad \qquad or $(S,T)$ is balanced\footnotemark and has small expansion. } \label{eqn:part}
\end{gather}
\footnotetext{Here we say a partition $V = S \sqcup T$ is {\em balanced} if $|S|,|T| \in [|V|/4,3|V|/4]$.}  
The above algorithm is a simple recursive application of the spectral partitioning algorithm from Cheeger's inequality (see Lemmas \ref{lem:cheeger}, \ref{lem:bisect} for more details). Note that the above algorithm may either output a large set which induces an expander (in which case we are done), or a balanced partition (say $\cP_0$) with small expansion. How can we proceed if the latter is the case? Following an idea from \cite{GR17}, we again apply the spectral partitioning (i.e, \eqref{eqn:part}) to each set in the partition in partition $\cP_0$ to construct a refinement of the partition, say $\cP_1$. Again, if $\cP_1$ contains a linear sized subset which induces an expander, then we are done -- otherwise, we again keep repeating the above process. We iteratively continue constructing a sequence of refinements $\cP_0 \subseteq \cP_1 \subseteq \cdots \subseteq \cP_t$ until one of the partitions contains a linear sized set which induces an expander. But then one can ask that how can we guarantee that the process terminates? This is where the {\it higher order Cheeger's inequality} (Theorem \ref{thm:h-cheeger}) comes to the rescue i.e., we show that if the process continues beyond some iteration $t = t(\delta)$, then $\cP_t$ is a balanced $K:=2^{t-1}$- partition of the vertex set. In particular, using Theorem \ref{thm:h-cheeger} and the fact that $G[V_{\rm good}]$ as low threshold rank, we can show that at least one of the $K$-sets in the partition must have large edge boundary i.e.,
\begin{equation}			\label{eqn:multi-1}
\max_{i \in [K]} \left|\partial_G(S_i)\right| \geq \Omega(\epsilon d n)~~~\textrm{(by an appropriate instantiation of parameters for \eqref{eqn:part}.)}
\end{equation} 
On the other hand, note that since the algorithm proceeds beyond iteration $t$, it follows that for each application of \eqref{eqn:part} in each of the $t$ iterations, the spectral partitioning algorithm returns a non-expanding partition (using the ``or'' guarantee from \eqref{eqn:part}), and hence the fraction of edges crossing the various sets in the partition $\cP_t$ must be small i.e.,
\begin{equation}			\label{eqn:multi-2}
\sum_{i \in [K]} \left|\partial_G(S_i) \right| \leq O(\epsilon^2 d n),
\end{equation} 
which contradicts the upper bound on the expansion from \eqref{eqn:multi-1}. In summary, the above arguments taken together imply that the above process must terminate during some iteration $t' \leq t$, resulting in a subset of size at least  $2^{-t'}\cdot n = \tilde{\Omega}(n)$\footnote{Here $\tilde{\Omega}$ hides poly-logarithmic in $\delta$ factors.} which induces an expander.

{\bf Finding Many Low Rank Sets}. Now that we have an algorithm that find a $\Omega(n)$-sized set (say $S_1$) that induces an expander, the next step is to find many such vertex disjoint sets in the graph. This is easily achieved by deleting the first such subset $S_1$ recovered by the above algorithm, and then again running the above algorithm on the graph $G[V \setminus S_1]$ to recover a linear sized subset $S_2 \subseteq V \setminus S_1$ which again induces an expander in $G$. However, note that the induced sub-graph $G[V \setminus S_1]$ does not automatically inherit the structural properties of $G[V_{\rm good}]$ and hence, additional care is need to ensure that the above algorithm will still succeed on the smaller induced sub-graph $G[V \setminus S_1]$. In particular, we shall again need to establish an analogue of \eqref{eqn:multi-1} where we show that any balanced $K$-way partition $\cP$ of the smaller set $V \setminus S_1$ will still have one expanding set. This is done by showing that any balanced $K$-way partition of $V \setminus S_1$ can be carefully extended to a balanced $K$-way partition $\cP'$ of $V$ such that 
\begin{equation}				\label{eqn:multi-3}
\max_{S \in \cP} \left|\partial_{G[V \setminus S_1]}(S)\right| \gtrsim  \max_{S' \in \cP'} \left|\partial_{G[V]}(S')\right|,
\end{equation}  
Note that the above immediately implies the desired $K$-way expansion bound for $\cP$ since the latter term can again be lower bounded by combining the higher order Cheeger's inequality with the threshold rank guarantee of the full graph $G[V]$. We point out that establishing \eqref{eqn:multi-3} is precisely where the bounded degree assumption on the graph comes in handy. The above (i.e., \eqref{eqn:multi-3}), along with an appropriately tailored version of \eqref{eqn:multi-2} will allow us to establish that the algorithm will again find a linear sized subset $S_2 \subset V_1 \setminus S_1$ which induces an expander. Overall, we keep iteratively finding and removing linear sized subsets $S_2,S_3,\ldots,$ -- each of which induces an expander -- until only $o_{\delta}(1)$-vertices remain; this results in an almost\footnote{An almost partition of a set $[n]$ is a collection of disjoint sets whose union contain $(1 - o(1))$-fraction of the elements.} partition $\cP := \{S_i\}_{i \in [N]}$ of the vertex set where each of the subsets in partition has small edge boundary, is linear sized, and induces an expander in the full graph $G$.

{\bf Stitching the sets together}. Recall that our final objective is not to find an almost partition consisting of induced low threshold rank subgraphs, but to find one large $(1 - o_{\delta}(1))$-subset $V'$ that induces a low threshold rank subgraph. To that end, we just show that the set $V' := \cup_{i \in [N]} S_i$ is itself such a set. To see this, observe that the adjacency matrix $A[V']$ of induced subgraph $G[V']$ is almost block diagonal (since the above step guarantees that only few edges cross the partition $\{S_i\}_{i \in [N]}$). Hence with some additional work we can conclude that the number of large eigenvalues in $A[V']$ must be at most the sum of number of large eigenvalues in each of the blocks $A[S_1],\ldots,A[S_N]$, each of which is again small on account of $G[S_i]$'s being expanders i.e., we can conclude ${\rm rank}_{1 - \delta^{O(1)}}(G[V']) \leq O(N)$. Furthermore, since each of the sets in the partition $\cP$ is linear sized, this establishes that $N$ is at most a constant (possibly depending on $\delta$), which implies that the threshold rank of $G[V']$ is at most $O_\delta(1)$.    

{\bf Reducing to the Bounded Degree Setting}. Lastly, we address the issue that in general the max degree of the underlying constraint graph can be arbitrarily large compared to the degree $d$ of the underlying good graph $G[V_{\rm good}]$. Towards that we introduce an additional pre-processing step which reduces the average degree of the remaining graph to $O(d)$ by deleting a small number of vertices, and then additionally deletes the vertices in the remaining subgraph which have degree larger than $d/\delta^{O(1)}$. For the first part, we use a $2$-factor approximation algorithm for the Partial Vertex Cover problem (Theorem \ref{thm:partial-vc}) that can be used to identify a small number of vertices that hits $\approx(d_{\rm avg}(G)) - d)n/2$ edges (where $d_{\rm avg}$ denotes the average degree). The subsequent deletion step again just removes a small number of vertices; this follows from a simple application of Markov's inequality. Finally, we remark that this again perturbs the spectral structure of the graph used that is used in the subsequent steps, and hence additional care is needed to make all of the above arguments go through.

%% file: prelims.tex
\section{Preliminaries}

Let $G=(V,E)$ be a graph with non-negative edge weights $w: E \to \Q^+$.
For a set $S \subset V$, let $\partial_{G}(S) \defeq \set{\set{(i,j) \in E}:\ i \in S, j \in V \setminus S } $.
The expansion of $S$ is defined as 
\[ \phi_G(S) \defeq \frac{\sum_{\substack{\set{i,j} \in \partial_G(S)}} w(\set{i,j})}{ \min \set{\sum_{i \in S}d_i, \sum_{i \in V \setminus S} d_i}},     \]
where $d_i \defeq \sum_{j \in V} w \paren{\set{i,j}}$.
The expansion of the graph $G$ is defined as $\phi_G \defeq \min_{S \subset V} \phi_G(S)$.

For a graph $G$, throughout we shall use $A_G$ to denote the normalized adjacency matrix of $G$, where $A_G(i,j) = \mathbbm{1}_{(i,j) \in E}/\sqrt{d_i,d_j}$ and let $L_G := I - A_G$ be the corresponding normalized Laplacian. For $i \in [n]$, we shall use $\lambda_i(L_G)$ to denote the $i^{th}$ smallest eigenvalues of $L_G$ respectively. For ease of notation, we shall sometime denote $\lambda_i(G) := \lambda_i(L_G)$. We shall also use $d_{\rm max}(G)$ and $d_{\rm av}(G)$ to denote the maximum and average degree of $G$. We will say a graph $G$ has bounded max degree if $d_{\rm max}(G) \leq O(1) \cdot d_{\rm av}(G)$.

\subsection{Spectral Partitioning Tools}

Our spectral decomposition algorithms crucially employ Cheeger type inequalities to relate the expansion profile of a graph to its spectral profile. Firstly, we shall need the following lemma which follows from the discrete Cheeger's inequality.

\begin{lemma}[\cite{Alon86,AM85}]	
\label{lem:cheeger}
	There exists a polynomial time algorithm that given a graph $G=(V,E)$ outputs a set $S$ such that $\phi_G(S) \leq \sqrt{2\lambda_2(\cL_G)}$.
\end{lemma}

We shall also need the following higher order variant of Cheeger's inequality.

\begin{theorem}[Higher Order Cheeger's Inequality \cite{LOT14}]\footnote{Also see \cite{LRTV12}.}			
\label{thm:h-cheeger}
	For any graph $G = (V,E)$ and a $k$-partition of the vertex set $V = S_1 \uplus S_2 \uplus \cdots \uplus S_k$ we have 
	\[
		\frac{1}{2} \lambda_k(L_G)\leq \max_{i \in [k]} \phi_G(S_i) \leq C \sqrt{\lambda_{2k} (L_G) \log k}  
	\]
	for some absolute constant $C$.
\end{theorem}

We shall also need the following $2$-approximation algorithm for partial vertex cover.

\begin{theorem}[Partial Vertex Cover \cite{BshoutyBurr98}]	
\label{thm:partial-vc}
Let $G = (V,E)$ be a graph on $n$ vertices. Suppose there exists a subset of vertices $V' \subset V$ which covers at least $t$-edges in $G$. Then, there exists a polynomial time algorithm which returns a set $V'' \subset V$ of size at most $2|V'|$ which covers at least $t$-edges in $G$.
\end{theorem}

\subsection{The Sum-of-Squares Hierarchy} 			\label{sec:lasserre}

The Sum-of-Squares (SoS) hierarchy~\cite{Shor87,Lass01} is a hierarchy of convex relaxations of a quadratic program constructed by adding increasingly stronger collections of constraints. In particular, the $R$-round SoS lifting of the SDP relaxation of a CSP $\Psi(V,E,\Sigma,\{\Pi_e\}_{e \in E})$ involves a degree-$R$ pseudo-distribution $\mu:= \{\mu_S\}_{S}$ which defines a local distribution $\mu_S$ for every subset $S \subset V$ of size at most $R$. These local distributions define distributions over assignments to variables corresponding to the subset. While the local distributions jointly may not correspond to a consistent distribution over all $n = |V|$-variables, the SoS constraints ensure that these local distributions must be ``locally consistent''. For e.g., for any choice of subsets $S,T \subseteq V$ such that $|S \cup T| \leq R$, and any event $\omega$ in the event space corresponding to the local distribution $\mu_{S \cap T}$, we have 
\[
\Pr_{\mu_S}\big[\omega\big] = \Pr_{\mu_T}\big[\omega\big] = \Pr_{S \cup T}\big[\omega\big].
\]
We refer interested readers to \cite{Laurent09} for more details on the SoS hierarchy and related topics. We introduce a few additional notation that will be used in this context. Throughout, given a degree-$R$ pseudo-distribution over labelings of $V$ (where $|V| = n$), we shall use $X_1,\ldots,X_n$ to denote the random variables denoting the label assigned by the pseudo-distribution. Again, the variables $(X_i)_{i \in [n]}$ together might not respect a joint distribution, but for every subset $S$ of size at most $R$, we will treat the variables $(X_i)_{i \in S}$ as being distributed according to $\mu_S$ (this is again possible due to the local consistency enforced by SoS). Furthermore, for any subset $S$ of size at most $R$ and any choice of $\alpha \in [k]^S$, we use $X_S = \alpha$ to denote the event that the set of variables in $S$ are assigned the labeling $\alpha$. Extending the notation, we shall use $\mu|X_S = \alpha$ to denote the degree $R - |S|$ pseudo-distribution consisting of the local distributions $\{\mu_{T}|X_{S} = \alpha\}_{|T| \leq R - |S|}$. Note that the SoS constraints will ensure that under any such conditioning, these local distributions would again be consistent for sets of size at most $R - |S|$.

%% file: find-expander.tex
\section{Partitioning Low Threshold Rank Graphs with Outliers}			
\label{sec:find-expander}

In this section we prove the following theorem which states there exists an efficient algorithm for finding a large sized vertex induced low threshold rank graph in {\em almost} low threshold rank graphs.

\begin{theorem}[Low Threshold Rank Recovery]
	\label{thm:thr-graph}
	Let $G = (V,E)$ be a graph on $n$-vertices such that there exists a subset $V_{\rm good} \subseteq V$ of size at least $(1 - \delta)n$ which satisfies (i) $\grank{\lambda^*}(G[V_{\rm good}]) \leq K$ with $\delta \leq (\lambda^*)^{100}$ and $K \leq \delta^{-0.001}$ and (ii) $G[V_{\rm good}]$ is $d$-regular. Then there exists an efficient algorithm which outputs a set $V'' \subseteq V$ of size at least $(1 - O(\delta^{1/11}))n$ such that $\grank{\delta^{0.8}}(G[V'']) \leq  \tilde{O}(\delta^{-1/16}$. Additionally, we have $|E(G[V''])| \geq dn/8$.
\end{theorem}

The proof of the above goes through the following spectral decomposition result stated below.

\begin{theorem}[Spectral Decomposition]				\label{thm:thres-decomp}
	Let $G = (V,E)$ be a graph on $n$ vertices such that there exists a subset $V_{\rm good} \subseteq V$ satisfying 
	(i) ${\rm rank}_{\geq 1 -  \lambda^*}(G[V_{\rm good}]) \leq K$, 
	(ii) $G[V_{\rm good}]$ has maximum degree $d$, and
	(iii) $|V_{\rm good}| \geq (1 - \delta)n$, 
	where $K \leq \delta^{-0.001}$. Furthermore, suppose we are given a subset $V' \subseteq V$ such that $|V'| \geq (1 - \alpha)n$ satisfying $d_{\rm max}(G[V']) \leq C_0 d$. Let $\gamma$ and $\epsilon$ be such that the following inequalities are satisfied:
	\[
	\epsilon < \left(\frac{\gamma^2(\lambda^*)^2}{2^8K^{4}/\log K}\right)^2 \qquad\qquad \delta \leq \frac{1}{100\log K} \qquad\qquad \gamma \geq  32 K^2 \cdot \max\big\{\alpha,C_0\delta\big\}.
	\]
	Then there exists a polynomial time algorithm which on input $V'$ outputs a set $\cS$ of disjoint subsets of vertices such that every $S \in \cS$ satisfies the following:
	\begin{itemize}
		\item[(i)] {\bf Small Edge Boundary}: $|\partial_{G[V']}(S)| \leq C \gamma^{-1}K^2\sqrt{\epsilon}d'n$.
		\item[(ii)] {\bf Linear Size}: $|S| \geq \gamma n /(4K^2)$.
		\item[(iii)] {\bf Induced Expansion}: ${\rm rank}_{\geq 1 -  \epsilon}(G[S]) \leq 1$.
		\item[(iv)] {\bf Constant Density}: $|E_{G}[S]| \geq |S|d/4$.
	\end{itemize}
	Furthermore we have $\left|\cup_{S \in \cS} S \right| \geq (1 - 2\gamma - \alpha)n$.
\end{theorem}

In words, the above theorem roughly says the following: suppose $G = (V,E)$ is a graph which has a large vertex induced sub-graph $G[V_{\rm good}]$ that is regular and has low threshold rank. Then, there exists a polynomial time spectral decomposition algorithm that on input $V'$ (which induces a large bounded max degree subgraph) constructs an almost partition $\cS$ of the vertex set $V'$ such that every subset in the partition induces an expander and has small edge boundary in the induced subgraph $G[V']$.

The algorithm from the above theorem is the key ingredient in the algorithm for Theorem \ref{thm:thr-graph}. In particular, our algorithm for Theorem \ref{thm:thr-graph} will invoke Theorem \ref{thm:thres-decomp} with the following choice of parameters. 
\begin{align}
	&\gamma \gets 32\delta^{1/10}K^2 \qquad\qquad C_0 \gets \delta^{-1/10} \qquad \qquad  \epsilon \gets \delta^{0.81} \qquad\qquad\alpha \gets 2\delta^{1/10} . \label{eqn:params} 
\end{align}

In the proof of Theorem \ref{thm:thr-graph} we verify that the above setting of parameters indeed satisfies all the inequalities required for instantiating Theorem \ref{thm:thres-decomp}. We shall defer the proof of Theorem \ref{thm:thres-decomp} to Section \ref{sec:thres-decomp} for now and use it to complete the proof of Theorem \ref{thm:thr-graph}. 

{\bf Algorithm for Theorem \ref{thm:thr-graph}}. The algorithm for Theorem \ref{thm:thr-graph} is described as Algorithm \ref{alg:t-rank}. 

\begin{algorithm}[ht!]	
	\SetAlgoLined
	\KwIn{A graph $G = (V,E)$, and parameter $\delta\in (0,1)$}
	Let $d_{\rm av} = 2|E(G)|/|V|$ be the average degree \;
	\For{$d_0 = 1$ to $d_{\rm av}$}{
		Set $d' = d_0 /\delta^{1/10}$ and  $t := (d_{\rm av} - d(1 - \delta))n/2$\;
		\underline{\it (1) Construction of $V'$}: \\
		Let $V_{\rm del}$ be the output of the $2$-approximation algorithm for $t$-Partial Vertex Cover (Theorem \ref{thm:partial-vc}) when run on $G$			\label{step:vc}\; 	
		\If{$|V_{\rm del}| \geq 2\delta n$}{Skip\;}
		Let $V_0 \gets V \setminus V_{\rm del}$			\label{step:del-0}\;
		Let $V_1$ denote the set of vertices with degree larger than $d'$\;
		Set $V' \gets V_0 \setminus V_1$			\label{step:del}\;
		\underline{\it (2) Spectral Decomposition Step}: \\
		Set parameters $\gamma,C_0,\epsilon,\alpha$ as in \eqref{eqn:params}\;
		Run Algorithm \ref{alg:low-rank-decomp} on $G[V']$ as above and parameters $(\epsilon,\delta,C_0,\gamma,K)$. Let $\cS = \{S_1,\ldots,S_N\}$ denote the output of the algorithm			\label{step:call-1}\;
		\underline{\it (3) Combining Subgraphs}: \\
		Construct the induced subgraph $G':= G\left[\cup_{i \in [N]} S_i\right]$\;
		\If{${\rm rank}_{\geq 1 -  \epsilon}(G') \leq 2N$ and $|\cup_{S \in \cS} S| \geq (1 - 2\delta^{1/11})n$}
		{
			Return subset $V'' := \cup_{i \in [N]} S_i$\;
		}
	}
	\caption{FindLowThresh}
	\label{alg:t-rank}
\end{algorithm}
The above algorithm broadly consists of three main components which we describe below. For simplicity, assume that the algorithm knows the degree of $G[V_{\rm good}]$. Then the algorithm goes through the following steps:
\begin{itemize}
	\item[1.] {\bf Construction of $V'$}: In the first step we construct a subset $V'$ containing $(1 - o_\delta(1))$ fraction of vertices such that the maximum degree of $G[V']$ is at most $O_{\delta}(d)$. This is again done using $2$-steps. Firstly, we use the $2$-approximation algorithm for partial vertex cover from Theorem \ref{thm:partial-vc} to find a $O(\delta n)$-sized subset $V_{\rm del}$ which has $\approx (d_{\rm av}(G) - d)|V|$ edges incident on it. Deleting the edges ensures that graph induced on the remaining vertices in $V \setminus V_{\rm del}$ has average degree at most $2d$. Secondly, we delete all vertices with degree larger than $C_0 d$ to form the subset $V'$. Using Markov's inequality and the lower bound on $|V \setminus V_{\rm del}|$, it follows that $|V'| \geq (1 - o_\delta(1))n$ vertices and has max-degree $O_\delta(d)$.
	\item[2.] {\bf Spectral Decomposition}: In the second step, we run the algorithm from Theorem \ref{thm:thres-decomp} instantiated with $V'$ constructed in step 1 along with the parameters defined in \eqref{eqn:params} to compute an almost partition $\cS := \{S_1,\ldots,S_N\}$ of $V'$ such that each subset $S_i$  (i) induces an expander (ii) is linear sized (iii) has small edge boundary in $G[V']$.
	\item[3.] {\bf Output combined sub-graph}: The algorithm then outputs the combined induced subgraph $G':=G[\cup_{i \in [N]}S_i]$, note that since $S_i$'s are linear sized, $N$ is constant and therefore $G'$ is essentially a union of a constant number of expanders and therefore has low threshold rank. 
\end{itemize}
In the remainder of this section, we now formally prove Theorem \ref{thm:thr-graph} by analyzing the guarantees of Algorithm \ref{alg:t-rank}.

\subsection{Proof of Theorem \ref{thm:thr-graph}}
Recall that in the setting of the theorem, $d$ is the degree of the induced (regular) subgraph $G[V_{\rm good}]$. Since the algorithm iterates over all guesses of degree $d_0$, it suffices to show that the algorithm produces a set $V'$ satisfying the guarantees of theorem in the iteration corresponding to $d_0 = d$. Henceforth, the remainder of the proof will just focus on the iteration $d_0 = d$. 

{\bf Construction of $V'$}. We begin by observing that for $d_0 = d$, the algorithm sets $t = (d_{\rm av} - 2d(1 - \delta))n/2$, where recall that $d_{\rm av}$ denotes the average degree of the graph $G$. Note that by definition we have 
\[
\left|\Big\{e \in E \Big| e \not\subset V_{\rm good}\Big\}\right| = |E| - |E[V_{\rm good}]| = d_{\rm av} n/2  - d(1 - \delta) n/2 \geq t,
\]
and therefore there exists a subset of size at most $\delta n$ vertices that hits at least $t$ edges in $G$. Hence, in Line \ref{step:vc}, the Partial Vertex Cover algorithm from Theorem \ref{thm:partial-vc} returns a subset $V_{\rm del} \subseteq V$ of size at most $2\delta n$ that hits at least $t$-edges. Consequently, the sub-graph $G[V \setminus V_{\rm del}]$ obtained by deleting $V_{\rm del}$ (Line \ref{step:del-0}) has average degree at most $2d$, and hence using Markov's inequality we have $|V_1| \leq 2\delta^{1/10}n$. Therefore, the surviving subset $V' = V_0 \setminus V_1$ constructed in Line \ref{step:del} has size at least $(1 - 3\delta^{1/10})n$, such that induced sub-graph $G[V']$ has maximum degree at most $d' = d/\delta^{1/10}$.

{\bf Spectral Decomposition Step}. Now we verify that the setting of parameters in \eqref{eqn:params} satisfies the inequalities required for instantiating Theorem \ref{thm:thres-decomp}. To begin with, observe that 
\begin{equation}			\label{eqn:cons-1}
	\left(C\frac{\gamma^2(\lambda^*)^2}{C^2_0 K^4 \log K} \right)^2 = \left(C\frac{\delta^{2/10} (\delta)^{0.002}}{\delta^{-2/10} \log {1/\delta^{0.001}}} \right)^2 > \delta^{0.81} \geq \epsilon.
\end{equation}
Furthermore, note that since $K \leq 1/\delta^{100}$, it follows that 
\begin{equation}				\label{eqn:cons-2}
	\frac{1}{100\log K} \geq \frac{1}{100 \log \left({1/\delta^{0.001}}\right)} \geq \frac{10}{\log(1/\delta)} > \delta,
\end{equation}
whenever $\delta \leq 1/100$. Finally, note that our choice of parameter $\gamma$ satisfies:
\begin{equation}				\label{eqn:cons-3} 
	\gamma/K^2 \geq 32\max \left\{\delta^{1/10},\delta^{9/10}\right\} = 32\max\big\{\alpha,C_0\delta\big\}. 
\end{equation}
In summary, \eqref{eqn:cons-1}, \eqref{eqn:cons-2} and \eqref{eqn:cons-3} together show that the setting of parameters satisfies all the conditions from Theorem \ref{thm:thres-decomp}. Therefore, running the algorithm from Theorem \ref{thm:thres-decomp} on $G[V']$ with these parameters (in Line \ref{step:call-1}) returns a collection of sets $\cS = \{S_1,\ldots,S_N\}$ which satisfy properties (i) - (iv) in Theorem \ref{thm:thr-graph}. In particular, this implies that 
\begin{equation}				\label{eqn:p1}
	\left|\cup_{i \in [N]} S_i \right| \geq ( 1- 2\gamma)n -  |V \setminus V'| \overset{1}{\geq} (1 - K^2 \delta^{1/10})n - 3\delta^{1/10}n 
	\overset{2}{\geq} (1 - 2K^2 \delta^{1/10})n,
\end{equation}
where inequality $1$ uses the bound on $V \setminus V'$ argued above, and inequality $2$ follows using $K \geq 1$. Furthermore, using property (iv) of Theorem \ref{thm:thres-decomp} we have 
\begin{equation}				\label{eqn:p2}
	|E(G[\cS])| \geq \sum_{i \in [N]} |E(G[S_i])| \geq \sum_{i \in [N]} \frac{d|S_i|}{4} = \frac{d}{4}\left| \cup_{i \in [N]} S_i \right| \geq \frac{dn}{8}
\end{equation}
which establishes the lower bound on the fraction of edges in the induced subgraph. 

{\bf Combining the Subgraphs}. Towards concluding the proof, we argue that the induced subgraph $G' = G[V'']$ satisfies the threshold rank bound guaranteed by the theorem. To that end, we construct another graph $\hat{G} = (V'',\hat{E})$ by including the following edges in the edge set $\hat{E}$:
\begin{itemize}
	\item For every $S \in \mathcal{S}$ and every edge $(i,j) \in E[S]$, we include the edge in $\hat{E}$.
	\item For every edge $(i,j) \in E\left[\cup_{p \in [N]}S_p\right] \setminus \left(\cup_{p \in [N]} E[S_p]\right)$ we add self-loops to vertices $i$ and $j$.  
\end{itemize}

We point that in the above construction of $\hat{G}$, a vertex in $\hat{G}$ can have multiple self-loops incident on it. Note that the above construction has the effect of making the adjacency matrix more diagonally dominant, and hence, it should spectrally dominate the adjacency matrix of $G'$. This intuition is made formal in the following spectral comparison lemma from \cite{LRTV11} which compares the eigenvalues of the graphs $G'$ and $\hat{G}$.

\begin{lemma}[Lemma 3~\cite{LRTV11}]				\label{lem:eig-bound}
	Given a graph $G' = (V,E')$ the following holds for any subset $F \subset E'$. Let $\hat{G} = (V,\hat{E})$ be constructed by removing the edge in $F$ and adding self-loops for every edge $e \in F$. Let $A'$ and $\hat{A}$ be the normalized adjacency matrices of $G'$ and $\hat{G}$ respectively. Furthermore, let $\lambda_i(A')$ and $\lambda_i(\hat{A})$ denote the $i^{th}$ largest eigenvalues of $A'$ and $\hat{A}$. 
	Then for every $i \in [|V|]$ we have $\lambda_i(\hat{A}) \geq \lambda_i(A')$.
\end{lemma}

Now observe that since $\widehat{A}$ is block-diagonal, its eigenvalues are the union of eigenvalues from each block matrix. Therefore,
\begin{equation}				\label{eqn:rank}
	{\rm rank}_{\geq 1 -  \epsilon}(\widehat{A}) \leq \sum_{i \in [N]} {\rm rank}_{\geq 1 -  \epsilon}(\widehat{A}[S_i]) \leq N,
\end{equation}
where in the last step we use the fact that for every $i \in [N]$ we have 
\[
\lambda_2(\widehat{A}[S_i]) = \lambda_2({A'}[S_i]) < 1 - \epsilon. 
\]
Here the equality uses the observation that adding self-loops to a graphs does not change its spectral gap, and the inequality is due to property (iii) of Theorem \ref{thm:thres-decomp}. Now, since the graphs $G'$ and $\hat{G}$ satisfy the requirements of Lemma \ref{lem:eig-bound}, using Lemma \ref{lem:eig-bound} we get that 
\[
\lambda_{N+1}(A') \leq \lambda_{N+1}(\hat{A}) < 1 - \epsilon. 
\]
The above inequality immediately implies that ${\rm rank}_{\geq 1 - \epsilon}(A\left[\cup_{S_i \in \cS} S_i\right]) \leq N$. This along \eqref{eqn:p1} and \eqref{eqn:p2} establishes that the subset $V''$ satisfies the guarantees claimed by the theorem.

\section{Proof of Theorem \ref{thm:thres-decomp}}				\label{sec:thres-decomp}

The main algorithm for Theorem \ref{thm:thres-decomp} is described and analyzed in Section \ref{sec:low-rank-decomp}. The algorithm itself involves several sub-routines which we briefly describe below.

\begin{enumerate}
	\item The basic inner sub-routine here is the {\em RankBisection} algorithm (Algorithm \ref{alg:rank-bisect}) which either outputs a large set which induces an expander, or outputs a balanced partition with small expansion -- this step employs the harder side of the basic Cheeger's inequality (Lemma \ref{lem:cheeger})
	\item Building on the above, the outer sub-routine (Algorithm \ref{alg:find-low-rank}) repeatedly uses Algorithm \ref{alg:rank-bisect} to compute finer partitions until it finds a linear sized vertex induced expander with small outer expansion.
	\item Finally, the {\em LowRankDecomp} algorithm (Algorithm \ref{alg:low-rank-decomp}) -- which is the main algorithm for Theorem \ref{thm:thres-decomp} -- uses Algorithm \ref{alg:find-low-rank} repeatedly to find non-expanding linear sized vertex induced expanders, until only a small number of vertices remain.
\end{enumerate}

In the remainder of this section, we shall describe and analyze the algorithms described above in the order they are stated, since the correctness of an algorithm in the above sequence relies on the guarantees of the previous algorithms.

\subsection{Inner Subroutine}			\label{sec:rank-bisect}

We begin by describing and analyzing a basic subroutine (Algorithm \ref{alg:rank-bisect}) which either outputs a large vertex induced expander or an almost balanced cut with small edge boundary.
\begin{algorithm}[ht!]
	\SetAlgoLined
	\KwIn{Graph $G_0 = (V_0,E_0)$ of maximum degree $d'$, parameters $\epsilon$ and $K$} 
	Initialize $S_1 \gets V_0$ and $t \gets 1$\; 
	\While{$|S'_t| \geq \frac{3n}{4}$}
	{\label{step:if}
		\If{${\rm rank}_{\geq 1 -  \epsilon}(G_0[S'_t]) \leq 1$}			
		{
			Break\;
		}
		Let $(S'_{t + 1},T'_{t + 1})$ be the output of the algorithm from Lemma \ref{lem:cheeger}
		on $G_0[S'_{t}]$ such that $|S'_{t+1}| \geq |T'_{t+1}| $ and $\phi_{G_0[S'_{t}]}(S'_{t + 1}),\phi_{G_0[S'_{t}]}(T'_{t + 1}) \leq \sqrt{2\epsilon}$ 
		\label{step:cheeger}\;
		Update $t \gets t + 1$\;
	}
	Return bipartition $(S'_t,V_0 \setminus S'_t)$\;
	\caption{RankBisection}
	\label{alg:rank-bisect}
\end{algorithm}

The following lemma states the guarantee for the above algorithm.

\begin{lemma}				\label{lem:bisect}
	Given a graph $G_0 = (V_0,E_0)$  with maximum degree $d'$, Algorithm \ref{alg:rank-bisect} outputs a partition $(S',T')$ of $V_0$ such that $|S'| \geq |T'|$ satisfying either (i) $|S'| \geq 3|V_0|/4$ and ${\rm rank}_{\geq 1 -  \epsilon}(G_0[S']) \leq 1$ or (ii)  $|V_0|/4 \leq |T'| \leq |S'| \leq 3|V_0|/4$ and $|\partial_{G_0}(S')| \leq 2\sqrt{\epsilon} d' |V_0|$.
\end{lemma}

\begin{proof}
	Suppose the algorithm outputs partition $(S'_t,\bar{S'_t})$ such that $|S'_t| \geq 3|V_0|/4$. Then it must have exited the while loop by satisfying the ``if'' clause from Line \ref{step:if} i.e., the subset $S'_t$ must satisfy ${\rm rank}_{\geq 1- \epsilon}(G_0[S'_t]) \leq 1$, and hence the partition $(S'_t,V_0 \setminus {S'_t})$ satisfies item (i) of the lemma. Otherwise, the algorithm exits the while loop at some iteration $t$ with $|S'_t| \leq 3|V_0|/4$, in which case we will show that the bipartition output by the algorithm will satisfy item (ii), thus concluding the proof of the lemma.
	
	To that end, let $(S'_1,T'_1),(S'_2,T'_2),\ldots,(S'_t,T'_t)$ denote the sequence of nested bi-partitions considered by the algorithm such that $S'_1 \supseteq S'_2 \supseteq \cdots \supseteq S'_t $. Note that by definition the sets $T'_1,\ldots,T'_t$ are all disjoint and satisfy
	\[
	\biguplus_{j \in [t]}T'_j = V_0 \setminus S'_t.
	\]
	Now, since for every iteration $j \in [t]$, the algorithm stays inside the while loop, we must have $\lambda_2(G_0[S'_t]) \leq \epsilon$, and hence in Line \ref{step:cheeger}, the algorithm from Lemma \ref{lem:cheeger} returns a bipartition $S'_j = S'_{j + 1} \uplus T'_{j + 1}$ such that 
	\begin{equation}				\label{eqn:cut}
		\left|E_{G_0[S'_j]}(S'_{j + 1},T'_{j + 1}) \right| \leq \sqrt{2 \epsilon} d' |T'_{j+1}|,
	\end{equation}
	Therefore we can bound the number of edges crossing the cut $(S'_t,\bar{S'_t})$ as:
	\begin{align}
		\left|\partial_{G_0}(S'_t) \right| &= |E_{G_0}(S'_t, V_0 \setminus S'_t)| {=} \sum_{j \in [t]} |E_{G_0}(S'_t, T'_{j})| \non\\
		& \leq \sum_{j \in [t]} |E_{G_0}(S'_j, T'_{j})|	
		= \sum_{j \in [t]} \left|E_{G_0[S'_{j-1}]}(S'_j, T'_{j}) \right| & 
		\textrm{(Since  $S'_{j-1} = S'_j \uplus T'_j$)}		\non\\
		&\leq \sum_{j \in [t]} 2\sqrt{\epsilon}d' |T'_j| & 
		\textrm{(Using \eqref{eqn:cut})}		\non\\
		&\leq 2\sqrt{\epsilon}d'|V_0|, 			\label{eqn:cut-1}
	\end{align}
	where in the last step we use the observation that $T'_1,T'_2,\ldots,T'_\ell$ are vertex disjoint sets whose union is contained in $V_0$. Furthermore, we must also have $3n/4 \geq |S'_t| \geq \left|S'_{t - 1}\right|/2 \geq |V_0|/4$ where the last inequality is due to the observation that since the algorithm does not exit the while loop in iteration $t - 1$, we must have $|S'_{t - 1}| \geq 3|V_0|/4$. Using identical arguments, we also get that $|T'_{t-1}| \geq |V_0|/4$. This together with the bound from \eqref{eqn:cut-1} completes the proof of the lemma. 
\end{proof}

\subsection{Outer Subroutine for Finding Expanders}				\label{sec:find-low-rank}

Next we describe the outer sub-routine which repeatedly uses Algorithm \ref{alg:rank-bisect} to construct a sequence of nested partitions until one of the partition contains a set that induces an expander and has small edge boundary.

\begin{algorithm}
	\SetAlgoLined
	\KwIn{Graph $G = (V,E)$, subset $V_\ell \subset V'$, parameters $\lambda^*, K, \epsilon$}
	Initialize $t \gets \ceil{\log K} + 1$ and partition $\cP_1 \gets V_\ell$ \;
	\For{$i = 1$ to $t$}
	{
		Let $\cP_{i+1} \gets \emptyset $ be the empty set\;
		Let $\cP_{i} = \{S_1,\ldots,S_{N_i}\}$ be the partition constructed in iteration $i-1$\;
		\For{$j = 1$ to $N_i$}
		{
			Let $(S'_j,T'_j)$ denote the output of Algorithm \ref{alg:rank-bisect} on $G[S_j]$			\label{step:call}\;
			\If{${\rm rank}_{\geq 1 -  \epsilon} (G[S'_j]) \leq 1$}
			{
				Return $S'_j$\;	
			}
			\If{${\rm rank}_{\geq 1 -  \epsilon} (G[T'_j]) \leq 1$}
			{
				Return $T'_j$\;	
			}
			Otherwise $\cP_{i+1} \gets \cP_{i+1} \cup \{S'_j,T'_j\}$\;
		}
	}
	\caption{FindLowRankSet}
	\label{alg:find-low-rank}
\end{algorithm}

The following lemma formally states the guarantees of the above algorithm.

\begin{lemma}					\label{lem:find-low-rank}
	Let $G=(V,E)$ and let $V_{\rm good},V',\epsilon,\delta,K$ be as in the setting of Theorem \ref{thm:thres-decomp}. Let  $V_\ell \subseteq V'$ be any subset of size at least $\gamma n$ such that $|\partial_{G[V']}(V_\ell)| \leq \lambda^* \gamma dn/10^3$. Then Algorithm \ref{alg:find-low-rank} on input $V_\ell$ outputs a set $S \subseteq V_\ell$ of size at least $\gamma n/(4K^2)$ such that (i) ${\rm rank}_{\geq 1 -  \epsilon}({G[S]} )\leq 1$ and (ii) $|\partial_{G[V_\ell]}(S)| \leq 2\sqrt{\epsilon} d'n$.
\end{lemma}

Before we prove the above lemma, we shall state and prove some additional results that will be useful in its proof. We begin with the following lemma which transfers the $K$-way expansion guarantees of $G[V_{\rm good}]$ to the induced sub-graph $G[V_\ell]$ for any large enough subset $V_{\ell} \subset V'$.
\begin{claim}[Local induced $K$-way expansion]		
	\label{cl:cl1}
	Let $\{S_1,S_2,\ldots,S_K\}$ be a non-trivial $K$-partition of vertex set $V_\ell$ such that $|S_i| \geq \gamma n/4K^2$ for every $i \in [K]$. Then we have
	\[ \max_{i \in [K]} \left|\partial_{G[V']}(S_i)\right| \geq \frac{\lambda^* \gamma dn}{64K^2}. \]  
\end{claim}
\begin{proof}
	For brevity we denote $\gamma':= \gamma/(4K^2)$. Now, as a first step, let $T:= V' \setminus V_{\ell}$, and define sets $S'_1,\ldots,S'_K$ as 
	$S'_1 := S_1 \cup T_1$ and $S'_i := S_i$ for every $i \in \{2,\ldots,K\}$. Then, for every $i \in [K]$, using Claim \ref{cl:cont-1} we have 
	\begin{equation}				\label{eqn:st-1}
		\left|\partial_{G[V_\ell]}(S_i) \right| \geq \left| \partial_{G[V']}(S'_i) \right| - \left|\partial_{G[V']}(V_{\ell})\right|. 
	\end{equation}
	Next, for every $i \in [K]$, define the set $S''_i := S'_i \cap V_{\rm good}$. Note that the sets $S''_1,\ldots,S''_K$ now form a partition of the set $V' \cap V_{\rm good}$. Then using Claim \ref{cl:cont-2}, for every $i \in [K]$, we can further lower bound:
	\begin{equation}				\label{eqn:st-2}
		\left|\partial_{G[V']}(S''_i)\right| \geq \left|\partial_{G[V' \cap V_{\rm good}]} (S'_i \cap V_{\rm good})\right| =  \left|\partial_{G[V' \cap V_{\rm good}]}(S''_i)\right|.
	\end{equation}
	Finally, let $T':= V_{\rm good} \setminus V'$. Then we define the sets $\tilde{S}_1,\ldots,\tilde{S}_K$ as $\tilde{S}_1 := S''_1 \cap T'$ and $\tilde{S}_i := S''_i$ for every $i \in \{2,\ldots,K\}$. Then observe that $\tilde{S}_1,\ldots,\tilde{S}_K$ form a partition of the set $V_{\rm good}$. Furthermore, using Claim \ref{cl:cont-3}, for every $i \in [K]$ we can bound:
	\begin{equation}				\label{eqn:st-3}
		\left|\partial_{G[V' \cap V_{\rm good}]}(S''_i)\right| \geq \left|\partial_{G[V_{\rm good}]}(\tilde{S}_i)\right| - \left|E[V_{\rm good} \setminus V', V' \cap V_{\rm good}]\right|
	\end{equation}
	Therefore, stitching together the bounds from \eqref{eqn:st-1}, \eqref{eqn:st-2} and \eqref{eqn:st-3} we get that the following holds for every $i \in [K]$: 
	\begin{equation}				\label{eqn:st-4}
		\left|\partial_{G[V_\ell]}(S_i) \right| \geq \left|\partial_{G[V_{\rm good}]}(\tilde{S}_i)\right| - 
		\left|E[V_{\rm good} \setminus V', V' \cap V_{\rm good}]\right| - \left|\partial_{G[V']}(V_\ell)\right|.
	\end{equation}
	Now we bound the various terms in RHS of \eqref{eqn:st-4}. For the first term, we observe that for every $i \in [K]$, we have that 
	\[
	\tilde{S}_i \supseteq S''_i \supseteq S'_i \setminus (V'\setminus V_{\rm good}) \supseteq S_i \setminus (V' \setminus V_{\rm good}), 
	\] 
	and hence $|\tilde{S}_i| \geq |S_i| - |V' \setminus V_{\rm good}| \geq (\gamma' - \alpha)n \geq \gamma'n/2$. Furthermore, note that by construction the sets $\tilde{S}_1,\ldots,\tilde{S}_K$ form a $K$-partition of $V_{\rm good}$. Hence using Theorem \ref{thm:h-cheeger} along with the fact that $\lambda_{K}(G[V_{\rm good}]) \geq \lambda^*$, we get the following lower bound on the $K$-way expansion of the partition $\{\tilde{S}_i\}^K_{i = 1}$:
	\begin{equation}			\label{eqn:st-5}
		\max_{i \in [K]} \left|\partial_{G[V_{\rm good}]}(\tilde{S}_i)\right| \geq \frac{\lambda^*d}{2} \cdot \frac{\gamma'n}{2}  = \frac{\lambda^* d \gamma}{16K^2 n}.
	\end{equation}
	Furthermore, note that from the setting of the lemma $|\partial_{G[V']}(V_\ell)| \leq \lambda^*\gamma d n/10^3K^2$. Finally, since $|V_{\rm good} \setminus V'| \leq \delta n$, it follows that $|E[V_{\rm good} \setminus V' , V_{\rm good} \cap V']| \leq d \delta n$. Plugging in these bounds along with \eqref{eqn:st-5} into \eqref{eqn:st-4} we get that 
	\begin{align*}
		\max_{i \in [K]} \left|\partial_{G[V_\ell]}(S_i)\right| 
		\geq \frac{\lambda^*\gamma dn}{16 K^2} - \frac{\lambda^*\gamma dn}{1000 K^2} - d \delta n 
		\geq \frac{\lambda^*\gamma dn}{32K^2}, 
	\end{align*} 
	which completes the proof of the claim.
\end{proof}
Next, we need a lemma to bound the number of edges crossing the sets in any partition $\cP_i$. To that end, we introduce the following notation: for any partition $\cP$ of $V_{\ell}$, define $\partial_{G[V_\ell]}(\cP) \defeq \cup_{S \in \cP} \partial_{G[V_\ell]}(S)$.
\begin{claim}			
	\label{cl:cl2}
	If Algorithm \ref{alg:find-low-rank} doesn't exit before the end of iteration $t$, 
	then for every iteration $i \in [t]$, we have 
	\[
	\left|\partial_{G[V_\ell]}(\cP_{i+1}) \right| \leq \left|\partial_{G[V_\ell]}(\cP_i) \right| + 2\sqrt{\epsilon}d'n, 
	\] 
	and therefore, $\left|\partial_{G[V_\ell]}(\cP_t)\right| \leq 2t\sqrt{\epsilon}d'n$. Furthermore, for any iteration $i \in [t]$ and for any set $S \in \cP_i$ we have $\left|\partial_{G[V_\ell]}(S)\right| \leq 2\sqrt{\epsilon}d'n$.
\end{claim}
\begin{proof}
	For any $i \in [t]$, denote by $N_i = 2^{i-1}$ the number of sets in partition $\cP_i$. Now, given partition $\cP_i$ consisting of sets $\hat{S}_1,\ldots,\hat{S}_{N_{i}}$ (say), the partition $\cP_{i + 1}$ consists of sets $\{S'_j,T'_j\}_{j \in [N_i]}$, where for each $j \in [N_i]$, the sets $S'_j,T'_j$ form a partition of $\hat{S}_j$ constructed by running the algorithm from Lemma \ref{lem:bisect} on $G[\hat{S}_j]$ (in Line \ref{step:call} of Algorithm \ref{alg:find-low-rank}). Therefore, it follows that 
	\begin{align*}
		\left|\partial_{G[V_\ell]}(\cP_{i+1}) \right|
		&= \sum_{j \neq j'} \left|E_{G[V_\ell]}(\hat{S}_j,\hat{S}_{j'}) \right| + \sum_{j \in [N_i]} \left|E_{G[V_\ell]}(S'_j,T'_j) \right| \\
		&= \left|\partial_{G[V_\ell]} (\cP_i)\right|  + \sum_{j \in [N_i]} \left|E_{G[\hat{S}_j]}(S'_j,T'_j) \right| \\
		&\leq \left|\partial_{G[V_\ell]} (\cP_i)\right|  + \sum_{j \in [N_i]} \sqrt{2\epsilon} d'|T'_j| \\
		&\leq \left|\partial_{G[V_\ell]} (\cP_i)\right|  + \sqrt{2\epsilon} d'n, 
	\end{align*} 
	where the last step uses the observation that since the sets $\hat{S}_1,\ldots,\hat{S}_{N_i}$ are all disjoint, the sets $T'_1,\ldots,T'_{N_i}$ are also disjoint. This establishes the first part of the claim.
	
	For the second part, fix a partition $\cP_i$ and a set $S \in \cP_i$. Note that for any iteration $j \leq i$, there exists set $U_j \in \cP_j$ such that $S$ is contained in $U_j$. Furthermore, let $U_{j,1},U_{j,2}$ be the output of Algorithm \ref{alg:rank-bisect} on $U_j$ (see Line \ref{step:call} of Algorithm \ref{alg:find-low-rank}) such that $S \subseteq U_{j,1}$ for every $j$. Note that by definition we must have $U_{j} = U_{j-1,1}$ for every iteration $j$, and the sets $U_{1,2},\ldots,U_{i-1,2}$ are all disjoint. Then we can bound the edge boundary of $S$ in $G[V_{\ell}]$ as:
	\[
	\left|\partial_{G[V_\ell]}(S)\right| \leq \sum_{j= 1}^{i - 1} \left|E_{G[V_\ell]}(S,U_{j,2})\right| \leq \sum_{j = 1 }^{i  - 1} \left|E_{G[V_\ell]}(U_{j,1},U_{j,2})\right| \overset{1}{\leq} 2\sqrt{\epsilon}d'\sum^{i - 1}_{j = 1}|U_{j,2}| 
	\overset{2}{\leq} 2\sqrt{\epsilon}d'n
	\] 
	where step $1$ follows using Lemma \ref{lem:bisect} and step $2$ again follows using the fact that the sets $U_{1,2},\ldots,U_{i-1,2}$ are all disjoint.
\end{proof}
\begin{claim}		
	\label{cl:cl3}
	If Algorithm \ref{alg:find-low-rank} doesn't exit before the end of iteration $[t]$, 
	then for any iteration $i \in [t]$, for any set $S \in \cP_i$ we have 
	$(1/4)^i |V_\ell| \leq |S| \leq (3/4)^i |V_\ell|$.
\end{claim}
\begin{proof}
	We prove this by induction on the iteration number $i$. For $i = 0$, the claim follows trivially for the base partition $\cP_0 = V_\ell$. 
	Now assume that the bound holds for partition $\cP_i$ for some iteration $i \in [t]$. Let $\cP_{i+1}$ be the refinement of the partition $\cP_i$. Fix a set $S \in \cP_{i+1}$. Then $S \in \{S'_j,T'_j\}$ for some $j \in [|\cP_i|]$ where $\{S'_j,T'_j\}$ is the output of Algorithm \ref{alg:rank-bisect} on the set $S_j \in \cP_i$. Therefore, using the guarantee of Lemma \ref{lem:bisect} we get that 		
	$(1/4)^i |V_\ell|\leq (1/4) |S_j| \leq |S| \leq (3/4) |S_j| \leq (3/4)^i |V_\ell|$, 
	where the first and the last inequalities follow from the induction hypothesis. 
\end{proof}

{\bf Proof of Lemma \ref{lem:find-low-rank}.} Now we prove the correctness of Algorithm \ref{alg:find-low-rank}.

\begin{proof}[Proof of Lemma \ref{lem:find-low-rank}]
	We claim that in at least one of the iterations $i \in [t]$, at least one execution of the Algorithm \ref{alg:rank-bisect} (in Line \ref{step:call})  must return a partition $(S'_j,T'_j)$ satisfying item (i) from Lemma \ref{lem:bisect}. Note that this immediately implies the lemma since item (i) of Lemma \ref{lem:bisect} implies that ${\rm rank}_{\geq 1 - \epsilon}(G[S'_j]) \leq 1$, and using Claim \ref{cl:cl3} we have
	\[
	\min\{|S'_j|,|T'_j|\} \geq (1/4)^j |V_\ell| \geq (1/4)^{t}|V_\ell| \geq \frac{|V_\ell|}{(\log K)^2}. 
	\]
	Therefore, for contradiction, we may assume that for every iteration $i \in [t]$, every execution of Algorithm \ref{alg:rank-bisect} (in Line \ref{step:call}) on every set $S_j \in \cP_i$ returns a partition $S'_j,T'_j$ of $S_j$ which satisfies item (ii) from Lemma \ref{lem:bisect}. Now consider the partition $\cP_t$ constructed by the algorithm in the $t^{th}$ iteration. Note that\footnote{For ease of notation, we assume that $K$ is a power of $2$, so that $2^{\lceil \log K \rceil} = K$. This is without loss of generality as otherwise the proof proceeds with $|\cP_t| = K^2$ instead of $|\cP_t| = K$ -- our parameters are chosen with enough slack to allow for the arguments to go through as is.} $|\cP_t| = 2^{\lceil \log K \rceil} = K$, and denote the corresponding sets in $\cP_t$ as $\{S_1,S_2,\ldots,S_K\}$. Then, on one hand, we can use Claim \ref{cl:cl2} and our choice of $t = \ceil{\log K} + 1$ to bound the number of edges crossing the partition $\cP_t$: 
	\begin{equation}				\label{eqn:edge-bound1}
		\left|\partial_{G[V_\ell]}(\cP_{t})\right|  \leq  4 d'\log K \sqrt{\epsilon}n.
	\end{equation}
	On the other hand, $\cP_t$ is a valid $2^{\ceil{\log K}}$ -partition of $V_\ell$, such that for every $S \in \cP_t$ we have $|S| \geq (1/4)^{\ceil{\log_{2} K}} |V_\ell| \geq \gamma n/(4K^2)$. Hence using Claim \ref{cl:cl1} we have
	\begin{equation}				\label{eqn:edge-bound2}
		\max_{i \in [K]}\left|\partial_{G[V_\ell]}(S_i)\right| \geq \frac{\lambda^* \gamma d n}{64 K^2}.
	\end{equation}   
	Comparing bounds \eqref{eqn:edge-bound1} and \eqref{eqn:edge-bound2} we get that 
	\[
	4 \log K \sqrt{\epsilon}d'n \geq \left|\partial_{G[V_\ell]}(\cP_t)\right| \geq \max_{i \in [K]}\left|\partial_{G[V_\ell]}(S_i)\right| \geq \frac{\lambda^* K^{-2} \gamma dn}{64},
	\]
	which on rearranging gives us $\epsilon \geq \left(\lambda^* \gamma/2^8K^2\log K\right)^2$ which again contradicts the choice of $\epsilon$ in Theorem \ref{thm:thres-decomp}. 
\end{proof}

\subsection{Low Rank Decomposition Algorithm}				\label{sec:low-rank-decomp}

In this section, we describe and analyze Algorithm \ref{alg:low-rank-decomp} which takes as input a graph $G$ such that it contains a large induced low threshold rank graph, and outputs an almost partition of the vertex set into linear sized expanders with small edge boundary. 
\begin{algorithm}[ht!]
	\SetAlgoLined
	\KwIn{Graph $G = (V',E')$ of maximum degree $d'$, parameters $\lambda^*, K, \epsilon$}
	Initialize $V_1 \gets V'$ and $\ell \gets 1$\; 
	\While{$|V_\ell| \geq 2\gamma n$}
	{
		Run Algorithm \ref{alg:find-low-rank} on $V_\ell$ instantiating it with $\gamma' = \gamma$;
		Let $S_\ell \subseteq V_\ell$ be the set returned by this algorithm such that 
		$|S_\ell| \geq \frac{\gamma n}{4K^2}$ and ${\rm rank}_{\geq 1 -  \epsilon}(G[S_\ell]) \leq 1$\; 
		Update $\mathcal{S} \gets \mathcal{S} \cup \{S_t\}$\;
		Update $V_{\ell +1} \gets V_{\ell} \setminus S_{\ell}$ and $\ell \gets \ell + 1$\;
	}
	Output the collection of sets in $\mathcal{S}$.
	\caption{LowRankDecomp}
	\label{alg:low-rank-decomp}
\end{algorithm}
The following lemma formally states the guarantee of Algorithm \ref{alg:low-rank-decomp}. 
\begin{lemma}				\label{lem:rank-decomp-corr}
	Let $G  = (V,E)$ be a graph with subsets $V_{\rm good},V'$ as in Theorem \ref{thm:thres-decomp}. Furthermore, suppose the parameters $d',\epsilon,\gamma,K,\delta$ satisfy the conditions from Theorem \ref{thm:thres-decomp}. Then Algorithm \ref{alg:low-rank-decomp} on input $V'$ outputs disjoint sets $S_1,\ldots,S_N$ such that $|\cup_\ell S_\ell|	\geq (1 - 2\gamma-\alpha)n$, where for every $\ell \in [N]$, the subset $S_\ell$ satisfies the following:
	(i) $|S_\ell| \geq \gamma n/(4K^2)$,
	(ii) ${\rm rank}_{\geq 1 -  \epsilon}(G[S_\ell]) \leq 1$, 
	(iii) $\left|\partial_{G[V']}(S_\ell)\right| \leq O(K^2\gamma^{-1}\sqrt{\epsilon}d'n)$ and
	(iv) $|E_{G}[S_\ell]| \geq d|S_\ell|/4$. 
\end{lemma}
{\bf Proof of Lemma \ref{lem:rank-decomp-corr}}. For any iteration $\ell$, denote  $E_{{\rm del},\ell} := \cup_{j \leq \ell} \partial_{G[V']}(S_j)$. We will establish the correctness of Algorithm \ref{alg:low-rank-decomp} in a couple of steps. To begin with, the first claim bounds the number of edges crossing the sets collected till iteration $t$.
\begin{claim}			\label{cl:edge-bound}
	For any iteration $\ell$ we have $|E_{\rm del,\ell}| \leq |E_{\rm del,\ell-1}| + 2\sqrt{\epsilon} d' n$. Consequently $|E_{{\rm del},\ell}| \leq 2\sqrt{\epsilon}d'\ell n$
\end{claim}
\begin{proof}
	By definition of $E_{{\rm del},\ell}$ we have
	\begin{align*}
		\left|E_{{\rm del},\ell}\right| 
		&=  \left|\cup_{j \leq \ell} \partial_{G[V']}(S_j)\right| \\
		&\overset{1}{\leq}  \left|\cup_{j \leq \ell - 1} \partial_{G[V']}(S_j)\right| + \left|\partial_{G[V_\ell]}(S_\ell)\right|  \\
		&\leq \left|E_{{\rm del},\ell - 1}\right| + \left|\partial_{G[V_\ell]} (S_{\ell}) \right|  \\
		&\overset{2}{\leq} \left|E_{{\rm del},\ell - 1}\right| + \sqrt{2 \epsilon} d' n,
	\end{align*}	
	where in step $1$ we use the observation that for any edge $e \in E_{{\rm del},\ell} \setminus E_{{\rm del},\ell - 1}$, the edge must be completely contained in $V_\ell$ and one of its end points must be in $S_\ell$. For step $2$, we use the bound on the size of the edge boundary of $S_{\ell}$ from Lemma \ref{lem:find-low-rank}. This establishes the first part of the lemma. The second part of the lemma follows by applying the bound repeatedly.
\end{proof}
The  next lemma shows that as long as $V_\ell$ is large enough, the execution of the while loop will return a linear sized expander satisfying conditions (i)-(iv). 
\begin{lemma}			\label{lem:iter}
	If $|V_\ell| \geq 2\gamma n$, then Algorithm \ref{alg:find-low-rank} always finds $S_\ell \subseteq V_\ell$ satisfying conditions (i)-(iv) from Lemma \ref{lem:rank-decomp-corr}.
\end{lemma}
\begin{proof}
	We prove this by induction on the iteration number $\ell$. Assume that the claim holds up to some iteration $\ell - 1$ such that $|V_{\ell}| \geq 2\gamma n$. Using the induction hypothesis, for every iteration $ t < \ell$ we must have $|S_{t}| \geq \gamma n/(4K^2)$ and hence we must  have $\ell \leq 4K^2/\gamma$. Therefore using Claim \ref{cl:edge-bound} we get that  
	\begin{equation}				\label{eqn:e-bound}
	\left|\partial_{G[V']}(V_\ell)\right| 
	= \left|E\left(V_\ell,\cup_{j \leq \ell - 1} S_j \right) \right| \leq |E_{\rm del,\ell}| \leq (4K^2/\gamma)(\sqrt{\epsilon} d'n) 
	\leq \frac{1}{32}\left(\frac{\gamma \lambda^* dn}{16K^2}\right) 
	\end{equation}
	where the last step follows from the choice of $\epsilon$ in Theorem \ref{thm:thres-decomp}. Therefore the subset $V_\ell$ satisfies the premise of Lemma \ref{lem:find-low-rank} and hence using the guarantee of Lemma \ref{lem:find-low-rank}, step $3$ must return a set $S_\ell$ of size at least $\gamma n/4K^2$ such that (a) $\grank{\epsilon}(G[S_\ell]) \leq 1$ and (b)$|\partial_{G[V_\ell]}(S_\ell)| \leq 2\sqrt{\epsilon} d' n$. Note that this establishes properties (i) and (ii) of Lemma \ref{lem:rank-decomp-corr}. Next, towards establishing item (iii) of the lemma, we observe that
	\begin{equation}			\label{eqn:bound}
		\left|\partial_{G[V']}(S_\ell)\right| \leq \left|\partial_{G[V_\ell]}(S_\ell)\right| + \left|\partial_{G[V']}(V_\ell)\right|
		\leq 2\sqrt{\epsilon} d' n + (4K^2/\gamma)(\sqrt{\epsilon} d'n) \leq (8K^2/\gamma)(\sqrt{\epsilon} d'n),
	\end{equation}
	where in the second inequality, we bound the first term using item (b) from above, and the second term is bounded using \eqref{eqn:e-bound}; this establishes item (iii). Finally for establishing item (iv) of the lemma, we observe that for any $S \in \mathcal{S}$ we have 
	\begin{eqnarray*}
		|E_{G[V']}[S]| &\geq& \left|\left\{e \in E \Big| e \cap S \neq \emptyset \right\} \right| -  \left|\partial_{G[V']}(S)\right|  \\
		&\overset{1}{\geq}& \frac{d}{2}|S \cap V_{\rm good}| - \left|\partial_{G[V']}(S)\right| \\
		&\overset{2}{\geq}& \frac{d}{2}\left(|S| - \delta n\right) - (8K^2/\gamma)(\sqrt{\epsilon} d'n) \\
		&\geq& \frac{d|S|}{4},
	\end{eqnarray*}
	where in step $1$ we bound the first term using following argument. Note that any vertex $i \in V_{\rm good}$ has $d$-edges incident on it and therefore, the number of edges incident on $S$ is at least $d |S \cap V_{\rm good}|/2$. In step $2$, lower bound the first term using $|V_{\rm good}| \geq (1 - \delta)n$ and upper bound the second term using condition (i). The final step follows using our choice of parameter $\epsilon$ and the lower bound on $|S|$ from condition (ii). 
	
\end{proof}

\subsection{Finishing the Proof of Theorem \ref{thm:thres-decomp}}

Using Lemma \ref{lem:iter}, we complete the proof of Theorem \ref{thm:thres-decomp}. Clearly, as long as $|V_\ell| \geq 2\gamma n$, Algorithm \ref{alg:find-low-rank} will find a set $S_\ell$ satisfying the conditions (i)-(iv), which are exactly the conditions (i)-(iv) from Theorem \ref{thm:thres-decomp}. Overall, the algorithm will output a collection of sets $\cS := \{S_1,\ldots,S_N\}$ such that $| [n] \setminus \cup_{i \in [N]} S_i  | \leq 2\gamma n$, this completes the proof of Theorem \ref{thm:thres-decomp}.

\subsection{Edge Boundary Inequalities}

\begin{claim}					\label{cl:cont-1}
	For every $i \in [K]$ we have 
	\[
	|\partial_{G[V_\ell]}(S_i)| \geq | \partial_{G[V']}(S'_i) | - |\partial_{G[V']}(V_{\ell})|
	\]	
\end{claim}
\begin{proof}
	Consider an edge $(a,b) \in \partial_{G[V']}(S'_i) \setminus \partial_{G[V_\ell]}(S_i)$. Then, since $S_i \subseteq V_\ell \subseteq V'$, we may assume that $a \in S'_i$ and $b \in V' \setminus S'_i$. Furthermore, since $(a,b) \notin \partial_{G[V_\ell]}(S_i)$ we must have either $a \notin S_i$ or $b \notin V_{\ell} \setminus S_i$ which along with the fact that $(a,b) \in \partial_{G[V']}(S'_i)$ implies that $\{a,b\} \not\subset V_\ell$. 
		
	On the other hand, we also claim that $\{a,b\} \not\subseteq V' \setminus V_{\ell}$. To see this, observe that since $V' \setminus V_{\ell} \subseteq S'_1$, if $\{a,b\} \subseteq V' \setminus V_{\ell}$, then $(a,b) \notin \partial_{G[V']}(S'_i)$ for any $i \in [K]$, which gives us a contradiction.
	
	The above observations together imply that the edge $(a,b)$ must cross the cut $V_{\ell},V' \setminus V_{\ell}$ i.e, $(a,b) \in \partial_{G[V']}(V_\ell)$. Since the above arguments hold for any edge $(a,b) \in \partial_{G[V']}(S'_i) \setminus \partial_{G[V_\ell]}(S_i)$, we have 
	\[
	|\partial_{G[V_\ell]}(S_i)| \geq |\partial_{G[V']}(S_i)| - |\partial_{G[V']}(V_\ell)|.	
	\]
\end{proof}
\begin{claim}				\label{cl:cont-2}
	For every $i \in [K]$ we have 
	\[
	|\partial_{G[V']}(S''_i)| \geq |\partial_{G[V' \cap V_{\rm good}]} (S'_i \cap V_{\rm good})|
	\]
\end{claim}
\begin{proof}
	Fix an edge $(a,b) \in \partial_{G[V' \cap V_{\rm good}]}(S'_i \cap V_{\rm good})$. Then without loss of generality, we have 
	\begin{align*}
		a \in S' \cap V_{\rm good} , b \in (V' \cap V_{\rm good}) \setminus (S'_i \cap V_{\rm good})
		& \Rightarrow a \in S'_i \cap V_{\rm good} , b \in V' \setminus S_i 		\tag{Since $S'_i \subseteq V'$}\\
		& \Rightarrow a \in V' , b \in V' \setminus S_i, 		\tag{Since $S'_i \subseteq V'$}
	\end{align*}
	which implies that $(a,b) \in \partial_{G[V']}(S'_i)$. Since the claim holds for every choice of $(a,b) \in \partial_{G[V' \cap V_{\rm good}]}(S'_i \cap V_{\rm good})$, the claim follows.
\end{proof}

\begin{claim}				\label{cl:cont-3}
	For every $i \in [K]$ we have 
	\[
	|\partial_{G[V' \cap V_{\rm good}]}(S''_i)| \geq |\partial_{G[V_{\rm good}]}(\tilde{S}_i)| - |E[V_{\rm good} \setminus V', V' \cap V_{\rm good}]|
	\]
\end{claim}
\begin{proof}
	Fix an edge $(a,b) \in \partial_{G[V_{\rm good}]}(\tilde{S}_i) \setminus \partial_{G[V' \cap V_{\rm good}]}(S''_i)$. Then without loss of generality,we have $a \in \tilde{S}_i$ and $b \in V_{\rm good} \setminus \tilde{S}_i$. Furthermore, we observe that
	\begin{align*}
		(a,b) \notin \partial_{G[V' \cap V_{\rm good}]}(S''_i)
		& \Rightarrow a \notin S''_i ~~ \vee ~~ b \notin (V' \cap V_{\rm good}) \setminus S''_i \\ 
		& \Rightarrow a \in \tilde{S}_i \setminus S''_i ~~ \vee ~~ b \in  S''_i 
	\end{align*}	
	However, note that $b \in S''_i$ is not possible, since then we would have $a,b \in \tilde{S}_i$ which contradicts $(a,b) \in \partial_{G[V_{\rm good}]}(\tilde{S}_i)$. Hence it must be the case that $(a,b) \in \tilde{S}_i \setminus S''_i \subseteq V_{\rm good} \setminus V'$. Furthermore, note that $(a,b) \in \partial_{G[V_{\rm good}]}(\tilde{S}_i)$ also implies that $\{a,b\} \not\subseteq V_{\rm good} \setminus V'$. Hence, putting together these observations implies that $(a,b) \in E[V_{\rm good} \setminus V',V' \cap V_{\rm good}]$. Since this claim holds for any choice of edge $(a,b) \in \partial_{G[V_{\rm good}]}(\tilde{S}_i)\setminus \partial_{G[V' \cap V_{\rm good}]}(S''_i)$, the claim follows.
\end{proof}

%% file: thm-main.tex
\section{Strong Unique Games in Almost Low Threshold Rank Graphs}

In this section, we prove Theorem \ref{thm:strong-ug}. 

\strongug*

The algorithm for the above theorem proceeds by pre-processing the graph by finding a large low-threshold rank subgraph (as guaranteed by Theorem \ref{thm:thr-graph}), following which it solves the $R$-level SoS relaxation subject to a couple of additional constraints which we describe in Figure \ref{fig:strug-constr}. Then it uses the low variance rounding to recover a large subset which is fully satisfiable. We describe the additional constraints and the algorithm below.

\begin{figure}[ht!]
	\begin{mdframed}
	\begin{itemize}
		\item {\bf Cardinality Constraint}. 
		\[
		\Ex_{i \sim V} \Pr_{X_i \sim \mu} \big[X_i = *\big] \leq 2\delta.
		\] 
		\item {\bf Edge Slack Constraint}. For every  $\forall (i,j) \in E, (a,b) \in [k] \times [k]$ such that $\pi_{i \to j}(a) \neq b$,
		\[
		\Pr_{(X_i,X_j) \sim \mu} \Big[X_i = a \wedge X_j = b\Big] = 0.
		\]
	\end{itemize}
	\end{mdframed}
	\caption{Additional Constraints for \stronguniquegames}
	\label{fig:strug-constr}
\end{figure}

	\begin{algorithm}[ht!]
		\SetAlgoLined
		\KwIn{A Unique Game instance $\cG(V_\cG,E_\cG,[k],\{\pi_{e}\}_{e \in E})$, and parameter $\delta \in (0,1)$}.
		Run Algorithm \ref{alg:t-rank} on $\cG$			\label{line:alg-call}\; 
		Let $V \subseteq V_\cG$ be the subset of vertices of size at least $(1 - \delta^{1/11})n$ with ${\rm rank}_{\geq 1 - \delta^{0.81}}(G) \leq K = \delta^{-1/10}(\log 1/\delta)^2$ be the subset of vertices guaranteed by Theorem \ref{thm:thr-graph}. Denote $E = E_\cG[V]$\;
		Set the number of SoS rounds to be $R = (k + 2)^2K+2$\;
		Solve the following $R$-level SoS lifting of the following SDP relaxation with $\Sigma = [k] \cup \{*\}$.  
				\begin{eqnarray*}
				\text{min}  &  \Ex_{(i,j) \sim E}\Pr_{(X_i,X_j) \sim \mu} \left[X_i \neq \pi_{j\to i}(X_j)\right]  & \\
				\text{s.t.} &  \mu~\textnormal{satisfies the constraints in Fig.\ref{fig:strug-constr}}
				\end{eqnarray*}		\label{line:sdp}\;
			Let $S \subseteq V$ be the set of size $(k+2)^2K$ guaranteed by Corollary \ref{corr:var-red}\;
			Sample assignment $x_S \in ([k] \cup \{*\})^S$ according to the distribution $\mu_S$\;
			Compute the set $V' \subseteq V$ as
			\[
			V' \defeq \left\{ i \in V~\Big|~{\rm Var}_{\mu|X_S = x_S}\Big[X_i\Big] \leq 0.1 \mbox{ and } \Pr_{X_i \sim \mu|X_S = x_S}\left[X_i = * \right] \leq 0.1  \right\}
			\]
			For every $i \in V'$, let $\sigma(i) \in [k]$ be the unique label for which $ \Pr_{X_i \sim \mu_i | x_S}\left[X_i = \sigma(i) \right] \geq 0.9 $\;
			Output the set $V'$ with the labeling $\sigma:V' \mapsto [k]$\;
		\caption{Robust UG}
		\label{alg:strong-ug}
	\end{algorithm}

	\subsection{Proof of Theorem \ref{thm:strong-ug}}
	
	We begin by observing that the underlying constraint graph of $\cG$ satisfies the premise of Theorem \ref{thm:thr-graph}, and hence in Line \ref{line:alg-call} the algorithm from Theorem \ref{thm:thr-graph} returns a subset $V \subseteq V_{\cG}$ satisfying the properties guaranteed by theorem. We list the ones required by our analysis here for convenience:
	\begin{itemize}
		\item $|V| \geq (1 - O(\delta^{1/12}))n$.
		\item ${\rm rank}_{\geq 1 - \delta^{0.81}}(\cG[V]) \leq K$ where $K = = \delta^{-1/16}{\rm polylog}(1/\delta)$.
		\item $|E_\cG[V]| \geq dn/8$.
	\end{itemize}
	The remainder of the proof will focus on analyzing the performance of the algorithm on the induced sub-instance $\cG[V]$. The proof shall roughly proceed along the following lines:
	\begin{itemize}
		\item Firstly, we shall show that SoS relaxation with the modified constraints is feasible and has value at most $2\delta^{0.9}$ (Claim \ref{cl:sdp-bound}).
		\item Next, in Corollary \ref{corr:var-red}, we will transfer the conditioning-reduces-variance property of unique games to that \stronguniquegames, and use that to show the existence of subset $S$ with small conditional average variance.
		\item Finally, Lemma \ref{lem:set-bound} bounds the size of $V'$ and Lemma \ref{lem:alg-sat} shows that the labeling returned by the algorithm satisfies all the induced constraints in $\cG[V']$. 
	\end{itemize}

	\begin{claim}				\label{cl:sdp-bound}
		The optimal value of the SDP is at most $\nu \defeq 2\delta^{9/10}$.
	\end{claim}
	\begin{proof}
		In order to establish the claim, it suffices to construct a feasible assignment to the vectors for which the objective is at most $\nu$. Let $V_{\rm good} \subseteq V$ be as guaranteed in the statement of Theorem \ref{thm:strong-ug}. Let $\sigma:V_{\rm good} \mapsto [k]$ be the partial labeling which satisfies all the induced constraints in $\cG[V_{\rm good}]$. The consider the (deterministic) distribution which assigns point-mass to the following labeling:
		\[
		X^*_i = 
		\begin{cases}
			\sigma(i)	& \mbox{ if } i \in V_{\rm good} \cap V, \\
			*			& \mbox{ if } i \in V \setminus V_{\rm good}.
		\end{cases}
		\]
		Note that since the above is a valid assignment, it also corresponds to a valid degree-$R$ pseudo-distribution, say $\mu$. Now, we quickly verify that the pseudo-distribution is feasible i.e., it satisfies the constraints from Figure \ref{fig:strug-constr}. Firstly, using the definition of $\mu$ we have that
		\[
		\Ex_{i \sim V}\Pr_{X_i \sim \mu} \Big[X_i = *\Big] = \Ex_{i \sim V} \left[\mathbbm{1}_{\{X_i = *\}}\right] \leq \frac{|V \cap V^c_{\rm good}|}{|V|} \leq 2\delta.
		\]
		Furthermore, for any edge $(i,j) \in E[V']$ and any $(a,b) \notin \Pi_{ij}$ we have that 
		\[
		\Pr_{(X_i,X_j) \sim \mu} \left[(X_i,X_j) = (a,b)\right] = 0,
		\]
		since if $i \notin V_{\rm good}$ or $j \notin V_{\rm good}$, we have $(X^*_i,X^*_j) \notin \Sigma \times \Sigma$ and hence the above follows trivially. On the other hand, if $\{i,j\} \subseteq V_{\rm good} \cap V$, then we know that $(X^*_i,X^*_j) \in \Pi_{ij}$ and hence the event $(X_i,X_j) = (a,b)$ again happens with probability $0$. To conclude the proof, we bound the objective value attained by $\mu$ by 
		\begin{eqnarray*}
		\frac{1}{|E|}\sum_{(i,j) \in E} \Pr_{(X_i,X_j) \sim \mu} \left[X_i \neq \pi_{j\to i}(X_j)\right] 
		&=& \frac{1}{|E|}\sum_{i \in V \setminus V_{\rm good}} \sum_{j \in N_\cG(i)}\Pr_{(X_i,X_j) \sim \mu} \left[X_i \neq \pi_{j\to i}(X_j)\right]  \\
		&\leq& \frac{d' \delta n}{dn /8} = 8 \delta^{-1/10} \delta = 8 \delta^{9/10}
		\end{eqnarray*}
		since $d' \leq d/\delta^{-1/10}$ using the guarantee of choice of the vertex set $V'$ in step $8$ of the algorithm.
	\end{proof}
	We now state the following key lemma which will be used in our analysis.
	\begin{lemma}[\cite{BRS11}, Lemma 8.2] 				\label{lem:var-red}
		Consider a Unique Game instance $\cG'(V,E,[k],\{\pi_{e}\}_{e \in E})$ and let $\mu$ be a $R$-level pseudo-distribution (where $R \geq k^2m + 2$) for the Unique Games SDP which satisfies 
		\[
		\Ex_{(i,j) \sim E} \Pr_{(X_i,X_j) \sim \mu} \left[X_i \neq \pi_{j\to i}(X_j)\right]  \leq \nu.
		\]
		Then for every $m$ exists a polynomial time computable set $S \subseteq V$ of size $k^2 m$ which satisfies
		\[
		\Ex_{i \sim V} \left[{\rm Var}_\mu\Big[X_i|X_S\Big]\right] \leq \frac{C\nu}{\lambda_m},
		\] 
		where $C>0$ is a constant independent of all other parameters.
	\end{lemma}
	
	Using the above lemma we can establish the following corollary which guarantees the existence of a set for which the expected conditioning reduces average variance.
	
	\begin{corollary}				\label{corr:var-red}
		Let $\mu$ be the pseudo-distribution computed by the algorithm in Line \ref{line:sdp}. Then there exists a polynomial time computable set $S \subseteq V$ of size $(k+2)^2 m$ such that 
		\[
		\Ex_{i \sim V} \left[{\rm Var}_{\mu}\Big[X_i|X_S\Big]\right] \leq C \delta^{1/11}
		\] 
	\end{corollary}
	\begin{proof}
		Given $\cG$ we construct an {\em extended} \uniquegames~instance $\cG'(V,E,[k] \cup \{*,\bot\},\{\pi'_e\}_{e \in E})$ on the extended label set $\Sigma = [k] \cup \{*\} \cup \{\bot\}$ where $\bot$ can be thought of as a second dummy label. Here the only difference is in the set of constraints on the edges. For every edge $e = (i,j)$, we define the new constraint set $\Pi'_e: = \Pi_e \cup \{(\bot,*),(*,\bot)\}$ i.e., the new constraint may be satisfied by a pair of labels $(a,b)$, if either (i) $(a,b)$ is satisfies the Unique Game constraint $\Pi_{i,j}$ or (i) if $(a,b) \in \{(*,\bot),(\bot,*)\}$. Note that $\{\Pi'_e\}_{e \in E}$ are also Unique Game constraints, and in particular, $\cG'$ is an \uniquegames~instance.
		
		Let $\mu$-be the degree-$R$ pseudo-distribution for the SoS relaxation in Line \ref{line:sdp} (which we denote by ${\sf SoS}(\cG[V]))$. We now observe the following:
		\begin{enumerate}
			\item Consider the $R$-level SoS relaxation for the unique game instance $\cG'$ (denoted by ${\sf SoS}_R(\cG')$). We claim that $\mu$ can be trivially extended to a feasible pseudo-distribution $\mu'$ for ${\sf SoS}_R(\cG')$ by assigning $0$-mass to any assignment that assigns the label $\bot$ to any vertex.
			\item Furthermore, the value of the objective in $SoS_R(\cG')$ obtained using $\mu'$ is exactly the corresponding value induced in ${\sf SoS}_R(\cG)$ using $\mu$. This is due to the observation that for every $(i,j) \in E$, a random draw of $(X'_i,X'_j) \sim \mu'$ is identically distributed as a random draw of $(X_i,X_j) \sim \mu$. In particular, this implies that 
			\begin{align*}
				\Ex_{(i,j) \sim \mu'} \Pr_{(X'_i,X'_j) \sim \mu} \left[X'_i = \pi'_{j \to i}(X_j)\right] 
				& = \Ex_{(i,j) \sim \mu} \Pr_{(X_i,X_j) \sim \mu} \left[X_i = \pi'_{j \to i}(X'_j)\right] \\
				& = \Ex_{(i,j) \sim \mu} \Pr_{(X'_i,X'_j) \sim \mu} \left[X'_i = \pi'_{j \to i}(X_j)\right] \\
				& \geq 1 - \nu.
			\end{align*}
		\item Now, since the constraint graph is unchanged, we have that ${\rm rank}_{\geq 1- \delta^{0.81} }(\cG') \leq K$. Then, invoking Lemma \ref{lem:var-red}, we can find a set $S \subseteq V$ of size $k' := (k+2)^2K + 2$ such that 
		\begin{equation}				\label{eqn:var-bd}
		\Ex_{i \sim V} \left[{\rm Var}_{\mu'}\Big[{X}'_i|{X}'_S\Big]\right] \leq \frac{C\nu}{\lambda_K} \leq C\frac{\delta^{9/10}}{\delta^{0.81}} \leq \delta^{1/12}.
		\end{equation}
		\end{enumerate}
		Finally, again recall that $x'_S \sim \mu'_S$ is again identically distributed as $x_S \sim \mu_S$ and hence for any $i \in V$
		\[
		{\rm Var}_{\mu}\left[X_i| X_S\right]
		= \Ex_{\alpha \sim \mu_S}\left[ {\rm Var}_{\mu|X_S = \alpha}\left[X_i\right]\right]
		= \Ex_{\alpha' \sim \mu'_S}\left[ {\rm Var}_{\mu'|X'_S = \alpha}\left[X'_i\right]\right] = {\rm Var}_{\mu'} \left[X'_i| X'_S\right],
		\]
		where the second equality again holds due to our construction of $\mu'$. Since the above holds for any $i \in V$, combining the above with \eqref{eqn:var-bd} gives us the claim.
	\end{proof}
	
	\begin{lemma}						\label{lem:set-bound}
		With probability at least $0.9$, the set $S$ returned by the algorithm is of size at least $(1 - O(\delta^{1/12}))n$.
	\end{lemma}
	\begin{proof}
		Let $S$ be the set of size $(k+2)^2m + 2$ as given by Corollary \ref{corr:var-red}. We begin by observing that
		\[
		\Ex_{i \sim V} \left[{\rm Var}_\mu\Big[X_i|X_S\Big]\right] = \Ex_{x_S \sim \mu_S} \Ex_{i \sim V} \left[{\rm Var}_{\mu|X_S = x_S}\Big[X_i|X_S = x_S\Big]\right] \leq C \delta^{1/12}
		\]
		and therefore using Markov's inequality, with probability at least $0.9$ over the choices of $X_S$ we have 
		\[
		\Ex_{i \sim V} \left[{\rm Var}_{\mu|X_S = x_S}\Big[X_i\Big]\right] \leq 10C \delta^{1/12}
		\]
		Therefore, again by averaging, for at least $\left(1 -  O\left(\delta^{1/12}\right)\right)$-fraction of $i \in V$, we have ${\rm Var}_{\mu|X_S = x_S}[X_i] \leq 0.9$. On the other hand, we also have $\Ex_{i \sim V} \Pr_{\mu|X_S = x_S}[X_i = * ] \leq \delta$, and therefore by averaging, for at least $(1 - 10\delta)$-fraction of choices of $i \in V$, we have $\Pr_{\mu|X_S = x_S}[X_i = *] \leq 0.1$. Combining the two bounds gives us that with probability at least $0.8$ we have $|T| \geq (1 - O(\delta^{1/12}))|V|$. The claim now follows using the lower bound $|V| \geq (1 - O(\delta^{1/12}))n$.
	\end{proof}	
	
	Finally, the following lemma shows that the labeling $\sigma$ returned by the algorithm satisfies all the constraints induced on the set $S$.
	
	\begin{lemma}[Folklore]				\label{lem:alg-sat}
		Let $V'$ be the set output by the algorithm. Then for every $(i,j) \in E[V']$ we have $\pi_{i \to j}(\sigma(i)) = \sigma(j)$.
	\end{lemma}
	\begin{proof}
		We begin by showing that the labeling $\sigma$ is well defined. Fix any $i \in V'$. Then we know that ${\rm Var}[X_i|x_S] \leq 0.1$. Then by definition we have 
		\[
		\max_{a \in [k] \cup \{*\}} \Pr_{\mu|X_S = x_S}\Big[X_i = a \Big] \geq \sum_{a \in [k] \cup \{*\}} \Pr_{\mu|X_S = x_S}\Big[X_i = a \Big]^2 = 1 - {\rm Var}_{\mu|X_S = x_S}\Big[X_i\Big] \geq 0.9, 
		\]
		which implies that there exists a label (i.e., $\sigma(i)$) for which $\Pr_{\mu|X_S = x_S} \left[X_i = a\right] \geq 0.9$ -- note that such a label must also be unique since $\mu|X_S = \alpha$ is a valid distribution over assignments to the vertices in $S$. These arguments apply to any $i \in T$ and in particular imply that $\sigma$ is well-defined.
		
		We shall now show that consider any edge $(i,j) \in E[V']$. We now claim that $\pi_{i \to j}(\sigma(i)) = \sigma(j)$. We prove this by contradiction. Suppose not. By union bound we have 
		\begin{equation}				\label{eqn:lb}
		\Pr_{(X_i,X_j) \sim \mu|X_S = x_S}\Big[(X_i,X_j) = (\sigma(i),\sigma(j))\Big] \geq 1 - \Pr_{X_i \sim \mu|X_S = x_S}\Big[X_i \neq \sigma(i)\Big] - \Pr_{X_j \sim \mu|X_S = x_S}\Big[X_j \neq \sigma(j) \Big] \geq 0.8 
		\end{equation}
		On the other hand, since $\pi_{i \to j}(\sigma(i)) \neq \sigma(j)$, from the local slack constraints we have 
		\[
		\Pr_{(X_i,X_j) \sim \mu|X_S = x_S}\Big[(X_i,X_j) = (\sigma(i),\sigma(j)) \Big] = 0,
		\]
		which contradicts the lower bound from \eqref{eqn:lb}. Hence, the above arguments taken together imply that for every $(i,j) \in E[V']$ we must have $\pi_{i \to j}(\sigma(i)) = \sigma(j)$, thus establishing the claim. 
	\end{proof}	
	
	{\bf Proof of Theorem \ref{thm:strong-ug}} We put together the various results established above to complete the proof of Theorem \ref{thm:strong-ug}.
	\begin{proof}
		Note that the final set returned by the algorithm is $V'$. Lemma \ref{lem:set-bound} implies that the size of $V'$ is at least $(1 - O(\delta^{1/12}))n$. On the other hand, Lemma \ref{lem:alg-sat} implies that the partial labeling $\sigma:V' \to [k]$ constructed by the algorithm satisfies all the constraints induced in $\cG[V']$ i.e., the set $V'$ along with partial labeling $\sigma$ satisfies the guarantees of Theorem \ref{thm:strong-ug}.
	\end{proof}

%% file: hardness.tex
\section{Proof of Theorem \ref{thm:4lin}}

In this section, we prove Theorem \ref{thm:4lin} using the reduction from the \lbl~problem, which we describe below.

\begin{definition}[Bipartite Label Cover]				\label{defn:lbl}
	A \lbl~instance $\cL(U_{\cL},V_\cL,E_{\cL},[k],[s],\{\pi_e\}_{e \in E_\cL})$ is a constraint satisfaction problem characterized by a bi-regular bipartite graph on vertex sets $(U_\cL,V_{\cL})$ and a constraint (multi)-edge set $E_\cL$. Each edge constraint $e = (u,v)$ is identified with a projection map $\pi_{v \to u}: [k] \to [s]$. A labeling $(a,b) \in [s] \times [k]$ is said to satisfy an edge $(u,v)$ if $\pi_{v \to u}(b) = a$. The value ${\sf Val}(\cL)$ of the \lbl~instance is the maximum fraction of constraints in $\cL$ that can be satisfied by a labeling. 
\end{definition}

The following \NP-Hardness of Label Cover is well known and can be obtained by combining the PCP theorem with parallel repetition (see \cite{MR08} and references therein).

\begin{theorem}[\cite{MR08}]			\label{thm:lbl-hardness}
	The following holds for any constant $\epsilon > 0$. Given a \lbl~instance $\cL(V_\cL,E_{\cL}, [k],[s], \{\pi_e\}_{e \in E})$ (where $k = k(\epsilon), s = s(\epsilon)$, it is \NP-Hard to distinguish between 
	\[
	\textnormal{(i) YES Case: }{\sf Val}(\cL) = 1 \ \ \ \ \ \ \textnormal{ and } 
	\ \ \ \ \ \ \textnormal{(ii) NO Case: }{\sf Val}(\cL) \leq \epsilon
	\]
\end{theorem}

The main result of this reduction from the \lbl~problem to $4$-Lin, as stated by the following theorem.

\begin{theorem}					\label{thm:4lin-red}
	The following holds for any constant $\alpha \in (0,1)$. Given a \lbl~instance $\cL = (V_\cL,E_\cL,[k],\{\pi_{u \to v}\}_{(u,v) \in E_\cL})$, there exists a efficient procedure which outputs a $4$-Lin Instance $H = (V_H,E_H)$ such that the following properties hold.
	\begin{itemize}
		\item If $\cL$ is a $1$-satisfiable, then there exists an assignment $\sigma:V_H \to \{0,1\}$ which satisfies at least $(1 - 2\eta)$-fraction of constraints in $H$.
		\item If there exists a subset $S \subseteq V_H$ of size at least $\alpha|V_H|$ and an assignment $\sigma:S \to \{0,1\}$ which satisfies at least $(1/2 + \nu)$-fraction of the induced constraints in $E_H[S]$, then $\cL$ is $\Omega(\nu\alpha^{16}\eta^2)$-satisfiable.
	\end{itemize}
\end{theorem}

The above theorem combined with the \NP-Hardness of Theorem \ref{thm:lbl-hardness} gives us Theorem \ref{thm:4lin}. The theorem is established by combining the techniques of \cite{Has01} which was used to prove optimal inapproximability for $3$-Lin with new insights into the expansion properties of the local constraint graph induced by the dictatorship test for an edge.

\subsection{Preliminaries}

In this section, we review some necessary technical preliminaries used to prove Theorem \ref{thm:4lin}. Due to technical reasons, we shall need to reduce from the product variant of Label Cover which can be reduced from Label Cover, as stated in the following theorem.  
\begin{definition}					\label{defn:prod-lbl}
Given a \lbl~instance $\cL$, the product \lbl~instance $\cL^{\otimes 2}$ is an (multi) edge weighted instance of \lbl~defined as follows.
	\begin{itemize}	
		\item {\bf Variables}. The variable set of $V_{\cL^{\otimes 2}}$ is the right vertex set $V_{\cL}$.
		\item {\bf Constraints}. For every $u \in U_\cL$, and every pair of neighbors $v_1,v_2 \in N_{\cL}(u)$ we add a constraint $(v_1,v_2) \in E_{\cL^{\otimes 2}}$ identified with the constraints $\pi_{e,v_1}:= \pi_{v_1 \to u}$ and $\pi_{e,v_2}:= \pi_{v_2 \to u}$. 
	\end{itemize}
\end{definition}

Note that the constraints in the product \lbl~instance are of a different form compared to \lbl~as defined in Definition \ref{defn:lbl}; here the projections constraints check whether two labels from the larger label set $[k]$ project to the smaller label set $[s]$, whereas in the bipartite variant, constraints check if a label from the larger side $[k]$ projects to the label from the smaller set $[s]$ assigned on the other end point. To that end, we recall the following well reduction from \lbl~to its bipartite variant (for e.g., see \cite{FGRW12}).

\begin{theorem}[Folklore]					\label{thm:prod-lbl}
Given a \lbl~instance $\cL(V_{\cL},E_{\cL},[k],[s], \{\pi_{e}\}_{(u,v) \in e})$, there exists a polynomial time reduction which constructs a product \lbl~ instance $\cL^{\otimes 2}(V_{\cL},E_{\cL^{\otimes 2}},[k],\{\pi_{e,v}\}_{v \in e,e \in E_{\cL^{\otimes 2}}})$ satisfying the following properties.
	\begin{itemize}
		\item {\bf Completeness}: If ${\sf Val}(\cL) = 1$, then there exists a labeling $\sigma:V_{\cL} \to [k]$ such that 
		\[
		\Pr_{e = (v_1,v_2) \sim E_{\cL^\otimes 2}} \left[\pi_{e,v_1}(\sigma(v_1)) = \pi_{e,v_2}(\sigma(v_2))\right] = 1 .
		\]
		\item {\bf Soundness}: If ${\sf Val}(\cL) \leq \epsilon$ then for all labelings $\sigma:V_\cL \to [k]$ we have 
		\[
		\Pr_{e = (v_1,v_2) \sim E_{\cL^\otimes 2}} \left[\pi_{e,v_1}(\sigma(v_1)) = \pi_{e,v_2}(\sigma(v_2))\right] \leq \epsilon .
		\]
		\item {\bf Weak Expansion}: For any $S \subseteq V_\cL$, where say $|S| = \alpha|V_{\cL}|$, then we have $|E_{\cL^{\otimes 2}}[S]| \geq \alpha^2 |E_{\cL^{\otimes 2}}|$, where $|E_{\cL^{\otimes 2}}[S]|$ is the weight of constraints induced in the product \lbl~instance $\cL^{\otimes 2}[S]$.
	\end{itemize}
\end{theorem}
\begin{proof}
	The completeness and soundness of $\cL^{\otimes 2}$ follow from standard properties for the product graph, so all that is left is to establish weak expansion. Let $S$ be a subset of $V_\cL$ such that $|S| = \alpha|V_\cL|$. Note that due to the bi-regularity of the \lbl~instance $\cL$, the following is equivalent to sampling a random edge in $\cL^{\otimes 2}$:
	\begin{enumerate}
		\item Sample a random left vertex $u \in U_{\cL}$.
		\item Sample two random neighbors $v_1,v_2 \in N_{\cL}(u)$, and consider the edge constraint $e:=(v_1,v_2)$ with projections $\pi_{e,v_i} := \pi_{v_i \to u}$ for $i = 1,2$.
	\end{enumerate}
	Hence,the expected fraction of constraints induced by the set $S$ in the product \lbl~instance $\cL^{\otimes 2}$ can be expressed as 
	\begin{eqnarray}
	\Ex_{u \sim V_{\cL}}\left(\Ex_{v_1,v_2 \sim N_\cL(u)} \mathbbm{1}(\{v_1,v_2\} \subseteq S)\right)
	&=&\Ex_{u \sim V_{\cL}}\left(\Ex_{v \sim N_\cL(u)} \mathbbm{1}(v \in S)\right)^2 \\
	&\overset{1}{\geq}& \left(\Ex_{u \sim V_{\cL}}\Ex_{v \sim N_\cL(u)} \mathbbm{1}(v \in S)\right)^2 \\
	&\overset{2}{=}& \left(\Ex_{v \sim V} \mathbbm{1}(v \in S)\right)^2 = \alpha^2 
	\end{eqnarray}
	where step $1$ follows from Cauchy-Schwarz, and the next inequality uses the observation that since $\cL$ is bi-regular, the distribution induced by $u \sim U, v \sim N_{\cL}(u)$ is the uniform measure on $V$.
\end{proof}

\paragraph{Fourier Analysis and Long Codes.}

It is well known that the set of functions $f: \{0,1\}^k \to \{0,1\}$ forms a vector space with respect to using the usual notion of addition and multiplication of functions. Furthermore, the set of characters $\{\chi_\alpha\}_{\alpha \in \{0,1\}^k}$ forms an orthonormal basis for such functions where $\chi_\alpha$ is given by the map $x \mapsto (-1)^{\langle \alpha , x \rangle_{\mathbbm{F}_2}}$ (here $\langle \cdot, \cdot \rangle_{\mathbbm{F}_2}$ denotes the inner product function over $\mathbbm{F}_2$). In particular, any function $f:\{0,1\}^k \to [-1,1]$ admits a multi-linear expansion
\[
f = \sum_{\alpha \in \{0,1\}^k}\wh{f}(\alpha)\chi_\alpha
\]
where $\wh{f}(\alpha) = \Ex_{x \sim \{0,1\}^k} f(x)\chi_\alpha(x)$ is the Fourier coefficient corresponding to character $\chi_\alpha$. The following proposition reviews some basic properties of Fourier expansions (See \cite{OD14} for an extensive treatment of this topic).
\begin{proposition}[\cite{OD14}]		\label{prop:fourier}
For any function $f:\{0,1\}^k \to [-1,1]$, its corresponding Fourier expansion $f = \sum_{S} \wh{f}(S)\chi_S$ satisfies the following properties.
\begin{enumerate}
	\item[(i)] {\em Parseval's Identity}: $\Ex_{x \sim {0,1}^s} f(x)^2 = \sum_{\alpha \in \{0,1\}^k} \wh{f}(S)^2$.
	\item[(ii)] The empty Fourier coefficient satisfies $\wh{f}(\emptyset) = \Ex_{x \sim \{0,1\}^k}[f(x)]$.	 
	\item[(iii)] For any $x, y \in \{0,1\}^k$, the Fourier characters satisfy $\chi_\alpha(x \oplus y) = \chi_\alpha(x)\chi_\alpha(y)$.
	\item[(iv)] Let $\{0,1\}^k_\eta$ denote the distribution over boolean strings, where each bit is drawn independently from ${\sf Bernoulli}(\eta)$. Then $\Ex_{\rho \sim \{0,1\}^k} \chi_{\alpha}(\rho) = (1 - 2\eta)^{|\alpha|}$.
	\item[(v)] Let $s < k$ and consider a projection map $\pi:[k] \to [s]$. Given $x \in \{0,1\}^s$, define $\pi(x) \defeq (x_{\pi(1)},x_{\pi(2)},\ldots,x_{\pi(k)})$. 
	Then for any $\alpha \in \{0,1\}^k$ we have	$\Ex_{x \sim \{0,1\}^s} \chi_{\alpha}(\pi(x)) = \Ex_{x \sim \{0,1\}^s} \chi_{\pi^{(2)}(\alpha)}(x)$ where $\pi^{(2)}$ is the vector supported on the coordinates which appear an odd number of times in $\pi({\rm supp}(\alpha))$.
\end{enumerate}
\end{proposition}

Additionally, we shall also need the notion of {\em long codes}. Given label set $[k]$, a long code maps $i \in [k]$ to a string of length $2^k$. Formally, the encoding of $i \in [k]$ is given by the $i^{th}$ dictator function $f = \chi_{e_i}$ (where $e_i$ is the $i^{th}$ standard basis vector supported only on the $i^{th}$ coordinate). The reduction in Theorem \ref{thm:4lin-red} will use the standard template of a dictatorship test which takes (supposed) long codes as input and tests whether these long codes encode a labeling of the vertices of $\cL^{\otimes 2}$ which satisfy a non-trivial fraction of the constraints. This is a well established template for proving hardness of approximation results for CSPs; we refer interested readers to \cite{Has01,Khot-longcode} for more details.

\subsection{The Reduction}

The reduction for Theorem \ref{thm:4lin} is given by the dictatorship test gadget described in Figure \ref{fig:pcp-4lin}, the test itself is a straightforward extension of H\r{a}stad's $3$-Lin test \cite{Has01} to the setting of $4$-Lin.

\begin{figure}[h!]
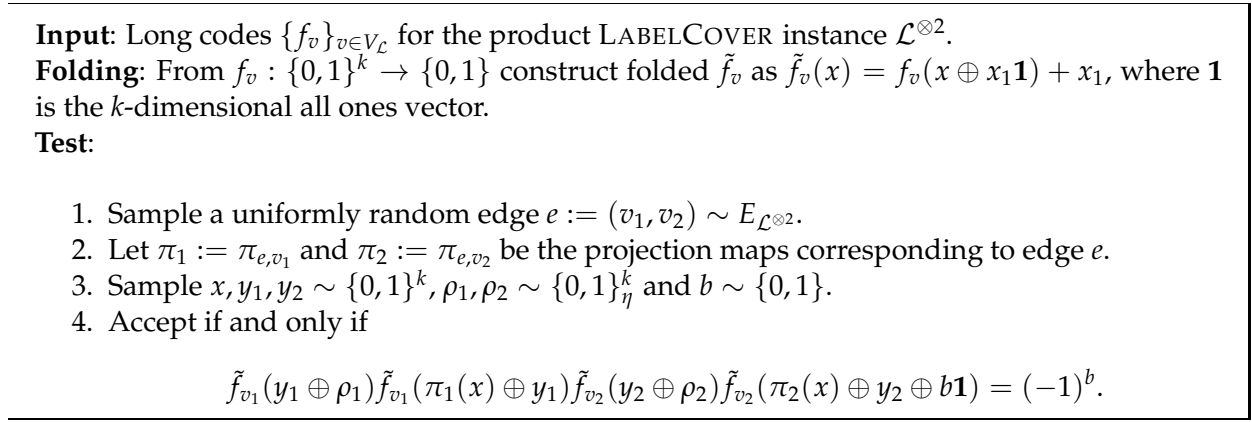
	
	\begin{mdframed}
		{\bf Input}: Long codes $\{f_v\}_{v \in V_\cL}$ for the product \lbl~instance $\cL^{\otimes 2}$. \\
		{\bf Folding}: From $f_v:\{0,1\}^k \to \{0,1\}$ construct folded $\tilde{f}_v$ as $\tilde{f}_v(x) = f_v(x \oplus x_1{\bf 1}) + x_1$, where ${\bf 1}$ is the $k$-dimensional all ones vector.\\
		{\bf Test}:	\\
		\begin{enumerate}
			\item Sample a uniformly random edge $e:= (v_1,v_2) \sim E_{\cL^{\otimes 2}}$.
			\item Let $\pi_1 := \pi_{e,v_1}$ and $\pi_2 := \pi_{e,v_2}$ be the projection maps corresponding to edge $e$.
			\item Sample $x,y_1,y_2 \sim \{0,1\}^k$, $\rho_1,\rho_2 \sim \{0,1\}^k_\eta$ and $b \sim \{0,1\}$.
			\item Accept if and only if 
			\[
			\tilde{f}_{v_1}(y_1 \oplus \rho_1) \tilde{f}_{v_1}(\pi_1(x) \oplus y_1)\tilde{f}_{v_2}(y_2 \oplus \rho_2)\tilde{f}_{v_2}(\pi_2(x) \oplus y_2 \oplus b{\bf 1}) = (-1)^b.\\
			\]
		\end{enumerate}
	\end{mdframed}
	\caption{$4$-Lin Dictatorship Test}
	\label{fig:pcp-4lin}
\end{figure}

The basic completeness and soundness properties as stated above can be established using techniques identical to the analysis for $3$-Lin in \cite{Has01}. We state them in the following theorem, and include a proof of it in Appendix \ref{sec:4Lin-prop} for the sake of completeness. 
\begin{theorem}
	The test described in Figure \ref{fig:pcp-4lin} satisfies the following properties.
	\begin{itemize}
		\item {\bf Completeness}. Suppose ${\sf Val}(\cL^{\otimes 2}) = 1$, then there exists long code tables $\{f_v\}_{v \in V_\cL}$ for which the test accepts with probability at least $1 - 3 \eta$.
		\item {\bf Soundness}. Suppose there exist long code tables $\{f_v\}_{v \in V_\cL}$ for which the test passes with probability at least $\frac12 + \nu$, then ${\sf Val}(\cL^{\otimes 2}) \geq \Omega(\nu \eta^2)$
	\end{itemize}
\end{theorem}

Establishing the $(1/2 + \epsilon)$-inapproximability on {\em every} constant sized induced sub-graph requires us to leverage the {\em weak expansion} properties of the outer and inner verifiers, which is new in this setting. For a vertex $v \in V_\cL$, and let $\cC_v$ be the set of variables identified with the positions in the long code table $f_v$. 
The following is the key technical lemma which will be used to strengthen the $4$-Lin Hardness.
\begin{lemma}[Weak Expansion in tests]				\label{lem:edge-bound}
	Fix a product \lbl~edge $(u,v)$ and let $A \subseteq \cC_u$, $B \subseteq \cC_v$ be such that $|A| = \alpha|\cC_u|$ and $|B| = \beta|\cC_v|$. Let $Q$ denote the random $4$-tuple supported on $\cC_u \cup \cC_v$ denoting they query indices of the test. Then 
	\[
	\frac{\alpha^2\beta^2}{2} \leq \Pr_{Q} \left[Q \subseteq A \cup B\right] \leq \alpha\beta .
	\]
\end{lemma}
\begin{proof}
Since the long code table is folded, we can identify the variables in $\cC_u$ with elements of $\{0,1\}^{k-1}$. In particular, when the verifier intends to query position $\theta \in \{0,1\}^k$ then (i) if $\theta_1 = 0$, then it queries the vertex corresponding to $\theta_{\geq 2}$ (ii) if $\theta_1 = 1$, then it queries the position $\theta_{\geq 2} \oplus {\bf 1}_{k-1}$. Now for $A \subseteq \cC_u$, define the function ${\bf 1}_A:\{0,1\}^k \to \{0,1\}$ as follows
	\[
	{\bf 1}_A(\theta) = 
	\begin{cases}
	\mathbbm{1}(\theta_{\geq 2} \in A) & \mbox{ if } \theta_1 = 0 \\
	\mathbbm{1}(\theta_{\geq 2} \oplus {\bf 1}_{k-1} \in A)  &\mbox{ if } \theta_1 = 1 
	\end{cases} .
	\]
	Note that by definition we have 
	\begin{eqnarray*}
		\wh{{\bf 1}_A}(\emptyset) = \E_{\theta \sim \{0,1\}^k} \Big[{\bf 1}_A(\theta)\Big] 
		&=& \frac12\Pr_{\theta \sim \{0,1\}^k}\left[{\bf 1}_A(\theta) = 1 \Big| \theta_1 = 0\right] + \frac12\Pr_{\theta \sim \{0,1\}^k}\left[{\bf 1}_A(\theta) = 1 \Big| \theta_1 = 1\right] \\
		&=& \frac12\Pr_{\theta \sim \{0,1\}^k}\left[\theta_{\geq 2}\in A\right] + \frac12\Pr_{\theta \sim \{0,1\}^k}\left[\theta_{\geq 2} + {\bf 1} \in A\right] \\
		&\overset{1}{=}& \Pr_{\theta \sim \{0,1\}^k}\left[\theta_{\geq 2}\in A\right] = \alpha.
	\end{eqnarray*}
	where in step $1$ we use the observation that $\theta_{\geq 2} + {\bf 1}$ is again uniformly distributed under $\{0,1\}^{k-1}$. Similarly we have $\wh{{\bf 1}_B}(\emptyset) = \beta$. For brevity, we denote $\tilde{y_i} = y_i \oplus \rho_i$, $z_1 = \pi_1(x) \oplus y_1$, and $z_2 = \pi_2(x) \oplus y_2 \oplus b{\bf 1}$.
	Now for the upper bound observe that 
	\begin{eqnarray*}
		\Pr_{Q} \left[Q \subseteq A \cup B\right] 
		&=& \Ex_{x,y,z} \left[\vone_A(\tilde{y}_1)\vone_A(z_1) \vone_B(\tilde{y}_2)\vone_B(z_2)\right] \\
		&\leq& \Ex_{x,y,z} \left[\vone_A(\tilde{y}_1) \vone_B(\tilde{y_2}_2)\right] \\
		&=& \Ex_{y_1} \left[\vone_A(\tilde{y}_1)\right] \Ex_{y_2} \left[\vone_B(\tilde{y}_2)\right]  =  \alpha \beta
	\end{eqnarray*}
	where the first equality is by definition, the middle inequality is using the fact that $\vone_A, \vone_B$ are indicator functions. Then penultimate step uses the independence of $\tilde{y}_1,\tilde{y}_2$ and the last equality is using the observation that marginally $\tilde{y}_1,\tilde{y}_2$ are uniformly random points combined with property (ii) of Proposition \ref{prop:fourier}. On the other hand,
	\begin{align*}
	&\Ex_{x,y,z} \left[\vone_A(\tilde{y}_1)\vone_A(z_1) \vone_B(\tilde{y}_2)\vone_B(z_2)\right] \\
	&= \Ex_{x,y,b,\rho}\left[\vone_A(\tilde{y_1})\vone_{A}(\pi_1(x) \oplus y)\vone_B(\tilde{y_2})\vone_{B}(\pi_2(x) \oplus y_2 \oplus b\cdot{\bf 1})\right] 
	\tag{Definition of $z_1,z_2$} \\
	&\geq\frac12 \Ex_{x,y,\rho}\left[\vone_A(\tilde{y_1})\vone_{A}(\pi_1(x) \oplus y)\vone_B(\tilde{y_2})\vone_{B}(\pi_2(x) \oplus y_2)\right] 
	\tag{Conditioning on $b = 0$}\\
	&=\frac12 \sum_{\alpha,\alpha',\beta,\beta'} \wh{\vone_A}(\alpha)\wh{\vone_A}(\alpha')\wh{\vone_B}(\beta)\wh{\vone_B}(\beta') \\
	&\qquad\times \Ex_{x,y,\rho}\Big[\chi_\alpha(\tilde{y_1})\chi_{\alpha'}(\pi_1(x) \oplus y)\chi_{\beta}(\tilde{y_2})\chi_{\beta'}(\pi_2(x) \oplus y_2)\Big] \tag{Fourier expansion of $\wh{\vone}_A,\wh{\vone}_B$} \\[8pt]
	&=\frac12 \sum_{\alpha,\alpha',\beta,\beta'} \wh{\vone_A}(\alpha)\wh{\vone_A}(\alpha')\wh{\vone_B}(\beta)\wh{\vone_B}(\beta')(1 - 2\eta)^{|\alpha|+ |\beta|} \\
	&\qquad\times \Ex_{x,y}\Big[\chi_\alpha({y_1})\chi_{\alpha'}(\pi_1(x) \oplus y_1)\chi_{\beta}(y_2)\chi_{\beta'}(\pi_2(x) \oplus y_2)\Big] \tag{Proposition \ref{prop:fourier} (iv)}	\\[8pt]
	&=\frac12 \sum_{\alpha,\beta} \wh{\vone_A}(\alpha)^2\wh{\vone_B}(\beta)^2(1 - 2\eta)^{|\alpha|+ |\beta|} 
	\Ex_{x}\Big[\chi_{\alpha}(\pi_1(x) )\chi_{\beta}(\pi_2(x))\Big] \tag{Since $\alpha \neq \beta \Rightarrow \langle \chi_{\alpha}, \chi_{\beta} \rangle = 0$ }\\
	&\geq \frac12 \sum_{\pi^{(2)}_1(\alpha) = \pi^{(2)}_2(\beta)} \wh{\vone_A}\left(\alpha\right)^2\wh{\vone_B}\left(\beta\right)^2(1 - 2\eta)^{|\alpha|+ |\beta|} 
	\tag{Proposition \ref{prop:fourier} (iii),(v)}\\
	&\geq \frac12 \wh{\vone_A}(\emptyset)^2\wh{\vone_B}(\emptyset)^2 = \frac{\alpha^2\beta^2}{2}
	\end{align*}
	where the last step follows from item $2$ of Proposition \ref{prop:fourier}.
\end{proof}

Using the above, we prove the following key lemma used for establishing inapproximability of Strong $4$-Lin.

\begin{lemma}				\label{lem:dense}
Let $H=(V_H,E_H)$ be the constraint multi-hypergraph output by the reduction in Figure \ref{fig:pcp-4lin}. Then for any $S \subseteq V_H$ such that $|S| \geq \alpha|V_H|$, we have $|E_H[S]| \geq \alpha^4|E_H|/16$.
\end{lemma}

\begin{proof}
	For every $v \in V_\cL$, let $\cC_v = \{f_v(x) : x\in \{0,1\}^k\}$ be the set of long code table variables corresponding to vertex $v$, and $S_v := \cC_v \cap S$. Furthermore, for any $S_1 \subseteq \cC_u$ and $S_2 \subseteq \cC_v$ we have 
	\[
	E_{H}(S_1,S_2) := \left\{ e \in E_H | e \cap \cC_u = S_1, e \cap \cC_v = S_2\right\}
	\]
	to be the set of constraints induced by $S_1 \uplus S_2$. Since $|S|/|V_H| \geq \alpha$ and $|\cC_u| = |\cC_v|$ for every $u,v \in V_{\cL^{\otimes 2}}$, we have $\Ex_{u \sim V_\cL}[|S_u|/|\cC_u|] \geq \alpha$. Define the set $V' = \{v \in V_{\cL} : |S_u| \geq (\alpha/2)|\cC_u|\}$. Then,
	\[
	\alpha \geq \Ex_{u \sim V_\cL}[|S_u|/|\cC_u|] \geq \frac{\alpha}{2}\Pr_{u \sim V_{\cL}}\left[u \notin V' \right] + \Pr_{u \sim V_{\cL}}\left[u \in V' \right]
	\]
	which on rearranging gives us that we have $|V'| \geq (\alpha/2)|V_{\cL}|$. Furthermore, from Lemma \ref{lem:edge-bound}, for every $(u,v) \in E[V']$ we have $|E_H(S_u,S_v)| \geq \alpha^2/8|E_H(\cC_u,\cC_v)|$. Therefore
	\begin{eqnarray*}
	|E_H[S]| = \sum_{(u,v) \in V_{\cL^{\otimes 2}}} |E_H(S_u,S_v)| 
	&\geq& (\alpha^2/8) \sum_{(u,v) \in E_{\cL^{\otimes 2}}[V']}|E_H(\cC_u,\cC_v)| \\
	&=& (\alpha^2/8) |E_{\cL^{\otimes 2}}[V']||E_H(\cC_u,\cC_v)| \\
	&\overset{1}{\geq}&  (\alpha^2/8)(\alpha^2/2)|E_{\cL^{\otimes 2}}||E_H(\cC_u,\cC_v)|  \\
	&=& (\alpha^4/16)|E_H|, 
	\end{eqnarray*}
	where in $1$, we use the weak expansion property of the instance $\cL^{\otimes 2}$ (Theorem \ref{thm:prod-lbl}). 
\end{proof}

\subsection{Proof of Theorem \ref{thm:4lin}}

The following gives the hardness for Strong $4$-Lin.

\begin{proof}
	Given $\cL$, let $H$ be the $4$-Lin instance output by the reduction from Figure \ref{fig:pcp-4lin}. The completeness direction follows directly from the analysis of the test in Figure \ref{fig:pcp-4lin}. For the soundness direction, suppose there exists $S \subseteq V_H$ such that $|S| \geq \alpha|V_H|$ satisfying at least $1/2 + \nu$-constraints in $E_H[S]$. We extend the labeling $\sigma$ to a labeling $\sigma':V_H \to \{0,1\}$ as follows. If $v \in S$, then $\sigma'(v) = \sigma(v)$. Otherwise assign $\sigma'(v)$ uniformly from $\{0,1\}$. Then the expected fraction of constraints satisfied by this labeling is at least 
	\[
	\frac{|E_H[S]|}{|E_H|}\left(\frac12 + \nu\right) + \left(1 - \frac{|E_H[S]|}{|E_H|}\right)\cdot\frac12 \geq \frac12 + \nu\frac{|E_H[S]|}{|E_H|} \geq \frac12 + \frac{\nu\alpha^4}{16}
	\]
	where the last step follows from Lemma \ref{lem:dense}. Therefore using the soundness analysis of the $4$-Lin test (Section \ref{sec:4Lin-prop}), there exists a labeling to vertices in $V_\cL $ which satisfies at least $\Omega(\nu^2\alpha^{16}\eta^2)$-fraction of constraints in $\cL$.
\end{proof}

%% file: appendix.tex
\section{Analysis of $4$-Lin Test}			\label{sec:4Lin-prop}

The completeness and soundness of the above test can be analyzed using techniques identical to \cite{Has01}.

{\bf Completeness}. Suppose there exists a labeling $\sigma: V_\cL \to [k]$ which satisfies all constraints in $\cL^{\otimes 2}$. Then for every $v \in V_\cL$, define the corresponding long code as $f_{v} := \chi_{\sigma(u)}$. Then it is easy to see that this assignment passes the test with probability at least $1 - 2\eta$.

{\bf Soundness}. Suppose the test passes with probability at least $1/2 + \nu$. Conditioned on the choice of the edge $(v_1,v_2) \in E_{\cL^\otimes 2}$, the acceptance probability of the test can be arithmetized as 
\begin{align*}
&\Pr\Big[\mbox{ Test Accepts }\Big] \\
&= \frac12 + \Ex_b(-1)^b\frac12\Ex_{x_i,y_i,\rho}\left[f_{v_1}(y \oplus \rho_1) f_{v_1}(\pi_1(x) \oplus y)f_{v_2}(y_2 \oplus \rho_2)f_{v_2}(\pi_2(x) \oplus y_2 \oplus b{\bf 1})\right]
\end{align*}
Conditioned on $b = 1$, the expectation term for a fixed choice of $(v_1,v_2)$ can be expanded as follows.
\begin{align*}
	&-\Ex_{x,y_i,\rho_i}\left[f_{v_1}(y_1 \oplus \rho_1) f_{v_1}(\pi_1(x) \oplus y_1)f_{v_2}(y_2 \oplus \rho_2)f_{v_2}(\pi_2(x) \oplus y_2 \oplus {\bf 1})\right] \\
	&= -\sum_{\alpha,\alpha',\beta,\beta' \in \{0,1\}^k} \wh{f_{v_1}}(\alpha)\wh{f_{v_1}}(\beta)\wh{f_{v_2}}(\alpha')\wh{f_{v_2}}(\beta') \\[-8pt]
	& \qquad\qquad\qquad\qquad\times 	\Ex_{x,y_i,\rho_i}\bigg[\chi_{\alpha}(y_1 \oplus \rho_1) \chi_{\beta}(\pi_1(x) \oplus y_1)\chi_{\alpha'}(y_2 \oplus \rho_2)\chi_{\beta'}(\pi_2(x) \oplus y_2 \oplus {\bf 1})\bigg]\\[8pt]
	&= -\sum_{\alpha,\alpha',\beta,\beta' \in \{0,1\}^k} \wh{f_{v_1}}(\alpha)\wh{f_{v_1}}(\beta)\wh{f_{v_2}}(\alpha')\wh{f_{v_2}}(\beta')(1 - \eta)^{|\beta| + |\beta'|}(-1)^{|\beta'|} \\[-8pt] 
	& \qquad\qquad\qquad\qquad\times \Ex_{x,y_i}\bigg[\chi_{\alpha}(y_1) \chi_{\beta}(\pi_1(x) \oplus y_1)\chi_{\alpha'}(y_2 )\chi_{\beta'}(\pi_2(x) \oplus y_2 )\bigg]\\[8pt]
	&= -\sum_{\beta,\beta' \in \{0,1\}^k} \wh{f_{v_1}}(\alpha)^2\wh{f_{v_2}}(\beta)^2(1 - \eta)^{|\beta| + |\beta'|}(-1)^{|\beta|}
	\Ex_{x}\left[\chi_{\alpha}(\pi_1(x))\chi_{\beta}(\pi_2(x))\right]\\
	&= -\sum_{\substack{\alpha,\beta \in \{0,1\}^k \\ \pi^{(2)}_1(\beta) = \pi^{(2)}_2(\beta')}} \wh{f_{v_1}}(\beta)^2\wh{f_{v_2}}(\beta')^2(1 - \eta)^{|\beta| + |\beta'|}(-1)^{|\beta'|}
\end{align*}
Similarly, for $b = 0$, we get that 
\begin{align*}
	&\Ex_{x,y_i,\rho_i}\left[f_{v_1}(y_1 \oplus \rho_1) f_{v_1}(\pi_1(x) \oplus y_1)f_{v_2}(y_2 \oplus \rho_2)f_{v_2}(\pi_2(x) \oplus y_2)\right] \\
	&=\sum_{\substack{\beta,\beta' \in \{0,1\}^k \\ \pi^{(2)}_1(\beta) = \pi^{(2)}_2(\beta')}} \wh{f_{v_1}}(\beta)^2\wh{f_{v_2}}(\beta')^2(1 - \eta)^{|\beta| + |\beta'|}
\end{align*}
Combining the expressions for the conditionings $b = 0$ and $b = 1$, we see that the terms with $|\beta|$ even cancel out, and therefore we can express the overall probability of the test accepting (over all edges) as follows.
\begin{align*}
	=\frac12 + \frac12 \Ex_{(v_1,v_2)} \sum_{\substack{|\beta| \textnormal{ is odd} \\ \pi^{(2)}_{v_1 \to u}(\beta) = \pi^{(2)}_{v_2 \to u}(\beta')}} \wh{f_{v_1}}(\beta)^2\wh{f_{v_2}}(\beta')^2(1 - \eta)^{|\beta| + |\beta'|}
\end{align*}
Now suppose the test passes with probability at least $1/2 + \nu$. Then,
\[
\Ex_{(v_1,v_2)} \sum_{\substack{|\beta| \textnormal{ is odd} \\ \pi_{v_1 \to u}(\beta) = \pi_{v_2 \to u}(\beta')}} \wh{f_{v_1}}(\beta)^2\wh{f_{v_2}}(\beta')^2(1 - \eta)^{|\beta| + |\beta'|} \geq 2\nu
\]

{\bf Randomized Decoding}. Now consider the following randomized labeling procedure. For every vertex $v \in V_{\cL}$, we do the following
\begin{enumerate}
	\item Sample $\alpha \in \{0,1\}^k$ with probability $\wh{f_v}(\alpha)^2$.
	\item Assign label $\sigma(v)$ uniformly from ${\rm supp}(\alpha)$.
\end{enumerate}
Note that step $1$ in the above procedure is well defined since using Parseval's identity (Proposition \ref{prop:fourier} (i)) we know that $\sum_{\alpha \in \{0,1\}^k} \wh{f_v}(\alpha)^2 = 1$ i.e, the $\{\wh{f_v}(\alpha)^2\}_{\alpha \in \{0,1\}^k}$ gives a well defined distribution over vectors in $\{0,1\}^k$. Finally, the expected fraction of edges satisfied can be lower bounded using standard steps as follows.
\begin{align*}
&\Ex_{e = (v_1,v_2)}\sum_{\substack{\alpha , \beta \\ \pi_{e,v_1}(\alpha) \cap \pi_{e,v_2}(\beta) \neq \emptyset}} \frac{\wh{f_{v_1}}(\alpha)^2\wh{f_{v_2}}(\beta)^2}{|\alpha||\beta|}  \\
&\geq \Ex_{e = (v_1,v_2)}\sum_{\substack{\pi^{(2)}_{e,v_1}(\alpha) = \pi^{(2)}_{e,v_2}(\beta) \\ |\alpha| \textnormal{ is odd}}} \frac{\wh{f_{v_1}}(\alpha)^2\wh{f_{v_2}}(\beta)^2}{|\alpha||\beta|} \\
&\geq \eta^2\Ex_{e = (v_1,v_2)}\sum_{\substack{\pi^{(2)}_{e,v_1}(\alpha) = \pi^{(2)}_{e,v_2}(\beta) \\ |\alpha| \textnormal{ is odd}}}\wh{f_{v_1}}(\alpha)^2\wh{f_{v_2}}(\beta)^2(1 - \eta)^{|\alpha| + |\beta|} \\
&\geq \eta^2 \nu
\end{align*}

\section{Hardness of General Strong $2$-CSPs}					\label{sec:ind-set-hard}

First, we show that without additional conditions on the CSP, recovering a satisfiable subset of vertices is at least as hard as \textsc{MaxIndSet}.
	
	\begin{observation}[Hardness of General Strong $2$-CSPs]
		Assuming $\P \neq \NP$ the following holds for any small $\epsilon > 0$. Given a $2$-CSP $\Psi(V,E,\{\psi\}_{e \in E})$ over label set $\{0,1\}$, it is $\NP$-Hard to find a subset $V' \subseteq V$ of size $|V'| \geq n^{1-\epsilon}|V^*|$ such that all induced constraints on $V'$ are satisfiable. Here $V^*$ is a set of largest cardinality for which there exists a labeling which satisfies all the induced constraints on $V^*$.
	\end{observation}
	\begin{proof}
		The proof of this uses the observation that the Max Independent Set problem (\textsc{MaxIndSet}) can be modeled as a Strong $2$-CSP. Indeed, let $G = (V,E)$ be a instances of \textsc{MaxIndSet}. We construct the $2$-CSP $\Psi$ as follows. The underlying constraint graph is $G$. Furthermore, for every edge $(u,v) \in E$, we add the constraints $x_u = x_v$ and $x_u \neq x_v$ to $\Psi$. This concludes the construction of $\Psi$. Now we claim that any subset $S \subset V$ which admits a labeling which satisfies all the induced constraints in $S$ must be an independent set. This follows from the observation that whenever a pair of vertices $(u,v)$ share an edge in $E$, no labeling can simultaneously satisfy both constraints between the pair $(u,v)$. Therefore, finding the largest satisfiable subset in $\Psi$ is exactly equivalent to finding the largest independent set in $G$. Now combining this with the following inapproximability of the \textsc{MaxIndSet} problem will conclude the proof.
		\begin{theorem}\cite{Zuck07}
			For all $\epsilon > 0$, it is \NP-Hard to approximate \textsc{MaxIndSet} to a factor of $n^{1 - \epsilon}$.
		\end{theorem}
	\end{proof}

%% file: vert-sep.tex
\section{Balanced Separator}

Here we shall prove the following Theorem.

\vertsep*

The proof of the above theorem uses the following general observation which connects the \balvertsep~problem to \stronguniquegames

\begin{proposition}					\label{prop:vert-sep}
	The \balvertsep~problem is an instance of \stronguniquegames~with cardinality constraints.
\end{proposition}

\begin{proof}
	Let $G = (V,E)$ be an instance of \balvertsep~with the guarantee that there exists a set $S \subseteq V$ of size at most $\delta n$ such that deleting $S$ disconnects the graph into two connected components, the smaller partition being of size $\gamma n$. Now we consider a \stronguniquegames~instance $\cG(V,E_{\cG},\{0,1\},\{\pi_e\}_{e \in E})$ on vertex set $V$. We add a constraint $\sigma(u) = \sigma(v)$ whenever $(u,v) \in E$. Furthermore, we add a global constraint that the number of $0$-labeled vertices is exactly $\delta n$.
	Now we observe the following:
	\begin{itemize}
		\item By construction, it follows that there exists a subset $S \subseteq V_\cG$ of size at most $\delta n$ such that $\cG[V\setminus S]$ is fully satisfiable using the labeling given by the partition into the two connected components, the smaller component being labeled as $0$.
		\item Conversely, let $\sigma:\tV \to \{0,1\}$ (where $|\tV| \geq (1 - \delta)n$ be a labeling which satisfies the global constraints and the induced edges constraints in $\cG[\tV]$. It is easy to verify that the labeling $\sigma$ must partition $\tV$ into two connected components with no edges in between such that the $0$-labeled set is of size $\gamma n$. Therefore $V \setminus \tV$ is a vertex separator of size at most $\delta n$.
	\end{itemize}
\end{proof}

Given the connection described by the above proposition, the algorithm for Theorem \ref{thm:sep} follows almost immediately. In particular, the algorithm is almost the same as \stronguniquegames~on alphabet size $2$, with added cardinality constraint which says that conditioned on any assignment, the fraction of $0$-labels is exactly $\delta$. Formally, the constraints and algorithm for Theorem \ref{thm:sep} are follows.

\begin{figure}
	\begin{mdframed}
		\begin{itemize}
			\item {\bf Cardinality Constraint}.,
			\[
			\Ex_{i \sim V} \Pr_{X_i \sim \mu} \left[X_i = *\right]  \leq 2\delta. 
			\] 
			\item {\bf Edge Slack Constraint}. For every  $\forall (i,j) \in E$, 
			\[
			\Pr_{(X_i,X_j) \sim \mu} \left[(X_i,X_j) \in \{(0,1),(1,0)\}\right] = 0.
			\]
			\item {\bf Partition Constraint}. For every $S \in \mathsmaller{{V \choose \leq r}}, \alpha \in \{0,1,*\}^S$,
			\[
			\Ex_{i \sim V} \Pr_{X_i \sim \mu} \left[X_i = 0\right] \in  \gamma \pm C \delta^{1/12},
			\]
			where $C$ is the constant from $O(\delta^{1/12})$ in Theorem \ref{thm:thr-graph}.
		\end{itemize}
	\end{mdframed}
	\caption{Additional Constraints for \balvertsep}
	\label{fig:sep-constr}
\end{figure}

	\begin{algorithm}[ht!]
		\SetAlgoLined
		\KwIn{A graph $G(V_G,E_G)$ and parameter $\delta$.}
		Construct the Unique Game $\cG(V_\cG,E_\cG,\{0,1\},\{\pi_e\}_{e \in E})$ as described in Proposition \ref{prop:vert-sep}\;
		Run Algorithm \ref{alg:t-rank} on $F$\; 
		Let $V \subseteq V_\cG$ be the subset of vertices of size at least $(1 - \delta^{1/10})n$ with ${\rm rank}_{1 - {\delta^{0.8}}}(G) \leq m = \delta^{-1/10}(\log 1/\delta)^2$ be the subset of vertices guaranteed by Theorem \ref{thm:thr-graph}. Denote $E = E_{\cG}[V]$\;
		Set the number of SoS rounds to be $R = 16m+2$\;
		Solve the following $R$-level SoS SDP relaxation:
		\begin{eqnarray*}
			\text{min}  &  \frac{1}{|E|}\sum_{(i,j) \in E} \Pr_{(X_i,X_j) \sim \mu} \left[X_i \neq X_j\right]& \\
			\text{s.t.} &  \mu \textnormal{ satisfies the constraints Fig. \ref{fig:sep-constr}.}
		\end{eqnarray*}\\
		Let $S \subseteq V$ be the set of size $16m$ guaranteed by Corollary \ref{corr:var-red} \;
		Sample assignment $x^S \in {0,1,*}^S$ according to the distribution $\mu_S$. \;
		Compute the set $V' \subseteq V$ as
		\[
		V' \defeq \left\{ i \in V~\Big|~{\rm Var}_{\mu|X_S = x_S}\Big[X_i\Big] \leq 0.1 \mbox{ and } \Pr_{X_i \sim \mu |X_S =  x_S}\left[X_i = * \right]\leq 0.1  \right\}.
		\]
		For every $i \in V'$, let $\sigma(i) \in \{0,1\}$ be the unique label for which $ \Pr_{X_i \sim \mu |X_S = x_S}\left[X_i = \sigma(i) \right] \geq 0.9 $\;
		Output the set $V'$ with the partitioning given by $\sigma:V' \mapsto {0,1}$\;
		\caption{Robust Vertex Separators}
		\label{alg:vertex-sep}
	\end{algorithm}
	
The correctness of the algorithm follows along the lines of Theorem \ref{thm:strong-ug}, along with a couple of additional observations:

\begin{enumerate}
	\item By combining the arguments from Proposition \ref{prop:vert-sep} and Claim \ref{cl:sdp-bound} we know that the SDP in Algorithm from \ref{alg:vertex-sep} is feasible and has optimal at most $2\delta^{9/10}$.
	\item Using Corollary \ref{corr:var-red}, we know that $\Ex_{i \sim V} [{\rm Var}[X_i|X_S]] \leq \delta^{1/10}$ which implies that with probability at least $0.9$, the random draw of $X_S = x_S \sim \mu_S$ will satisfy $\Pr_{i \sim V} [{\rm Var}[X_i|X_S= x_S] > 0.1] \leq O(\delta^{1/10})$. 
	\item Now, for every $i \in V$ define $p_a(i) := \Pr_{X_i \sim \mu|X_S = x_S} \left[X_i = a\right]$ for $a \in \{0,1,*\}$. Then the SoS constraints imply that
	\[
	\Ex_{i \sim V} \left[p_0(i) \right] \in \gamma \pm O(\delta^{1/12}).
	\]
	Using the above and the fact that $|V \setminus V'| \leq O(\delta^{1/12} n)$ (from Lemma \ref{lem:set-bound}) we have
	\begin{align*}
		\gamma - O(\delta^{1/12}) \leq \Ex_{i \sim V'}\left[ p_0(i)\right]
		&= \Ex_{i \sim V'}\left[ p_0(i) \cdot \mathbbm{1}_\{p_0(i) \geq 1 - \delta^{1/24}\}\right]
		+  \Ex_{i \sim V'}\left[ p_0(i) \cdot \mathbbm{1}_\{p_0(i) \leq \delta^{1/24}\}\right] \\
		&\leq \Ex_{i \sim V'}\left[  \mathbbm{1}_\{p_0(i) \geq 1 - \delta^{1/24}\}\right] + \delta^{1/24} \\
	\end{align*}
	which implies that the algorithm returns a labeling that labels $\gamma-O(\delta^{1/24})$-fraction of vertices as $0$.
	\item Finally  consider the partition $V' = A \uplus B$ given by the labeling $\sigma$. Using Lemma \ref{lem:alg-sat} we know that the labeling $\sigma$ satisfies all induced constraints in $\cG[V']$. This, along with the definition of $\cG$ in Proposition \ref{prop:vert-sep} implies that all the edges in $G[\tV]$ are either in $G[A]$ or $G[B]$ i.e., $A$ and $B$ are connected. This concludes the proof of Theorem \ref{thm:sep}.
\end{enumerate}

\section{Local-to-Global vs Conditioning Reduces Variance}				\label{sec:examp}

We point out that while {\em local-to-global} correlation is a property of the underlying constraint graph of the CSP, the stronger {\em Conditioning reduces Variance} type property required by our algorithms is predicate structure dependent. For instance consider the trivial predicate $\psi:[k] \times [k] \to \{0,1\}$ which is satisfied by all labelings i.e., ${\rm supp}(\psi) = [k] \times [k]$. Now consider a \mcsp~ on the complete graph on $n$-vertices with predicate $\psi$. Clearly the value of this CSP is $1$, and it is attained by any distribution on labels. Now consider the distribution on labels which assigns the uniform distribution on each vertex. This, as described above is an optimal distribution. However, note that conditioned on any assignment to any subset $S \subseteq [n]$, the distribution for the remaining vertices $[n] \setminus S$ is still the uniform distribution over labels. In particular, conditioning on a constant fraction of vertices does not reduce the average variance. On the other hand since the underlying constraint graph is the complete graph, the local-to-global correlation property is trivially satisfied here.